\documentclass{article}

\usepackage{fullpage}
\usepackage[all=normal,bibliography=tight]{savetrees}
\usepackage{amsmath}
\usepackage{amssymb,amsthm}
\usepackage{xspace}
\usepackage{comment}
\usepackage{tikz}
\usepackage{blindtext}
\usepackage{graphicx}

\usepackage{caption}
\usepackage{subcaption}
\newtheorem{theorem}{Theorem}
\newtheorem{lemma}[theorem]{Lemma}
\newtheorem{corollary}[theorem]{Corollary}
\newtheorem{definition}[theorem]{Definition}
\newtheorem{myclaim}[theorem]{Claim}

\newcommand{\Oh}{\mathcal{O}}

\newcommand{\maybeqed}{}

\def\cqedsymbol{\ifmmode$\lrcorner$\else{\unskip\nobreak\hfil
\penalty50\hskip1em\null\nobreak\hfil$\lrcorner$
\parfillskip=0pt\finalhyphendemerits=0\endgraf}\fi} 

\newcommand{\cqed}{\renewcommand{\qed}{\cqedsymbol}\maybeqed}

\newcommand{\shitarg}[1]{\textsc{$#1$-Subgraph Hitting}\xspace}
\newcommand{\cshitarg}[1]{\textsc{Colorful $#1$-Subgraph Hitting}\xspace}
\newcommand{\shit}{\shitarg{H}}
\newcommand{\cshit}{\cshitarg{H}}

\newcommand{\col}{\sigma}
\newcommand{\msep}{\mu}
\newcommand{\mnei}{\mu^\star}% {\eta}
\newcommand{\bd}{\ensuremath{\partial}}
\newcommand{\bdfun}{\lambda}
\newcommand{\labset}{\mathbb{L}}

\newcommand{\state}{\mathbf{s}}
\newcommand{\statefam}{\mathbb{S}}
\newcommand{\chunk}{\mathbf{c}}
\newcommand{\chunkfam}{\mathbb{C}}
\DeclareMathOperator{\inte}{int}
\newcommand{\labfun}{\Lambda}
\newcommand{\slice}{\mathbf{p}}
\newcommand{\slicefam}{\mathbb{P}}
\newcommand{\operadd}{\mathtt{enhance}}
\newcommand{\wG}{{\widehat{G}}}
\newcommand{\wX}{{\widehat{X}}}

\newcommand{\Tbag}{\beta}
\newcommand{\Tall}{\gamma}
\newcommand{\Tdown}{\alpha}
\newcommand{\Tree}{\ensuremath{\mathtt{T}}}

\begin{document}

\title{Hitting forbidden subgraphs in graphs of bounded treewidth\thanks{%
A preliminary version of this work has been presented at MFCS 2014.
The research leading to these results has received funding from the European Research Council under the European Union's Seventh Framework Programme (FP/2007-2013) / ERC Grant Agreement n. 267959 (the fourth author, while he was at University of Bergen, Norway), 
and n. 280152 (the second author), as well as OTKA grant NK10564 (the second author) and Polish National Science Centre grant DEC-2012/05/D/ST6/03214 (the first and the third authors, while they were at University of Warsaw, Poland).}}

\author{
  Marek Cygan\thanks{Institute of Informatics, University of Warsaw, Poland, \texttt{cygan@mimuw.edu.pl}.} \and
  D\'{a}niel Marx\thanks{Institute for Computer Science and Control, Hungarian Academy of Sciences (MTA SZTAKI), Hungary, \texttt{dmarx@cs.bme.hu}.} \and
  Marcin Pilipczuk\thanks{Department of Computer Science, University of Warwick, United Kingdom, \texttt{malcin@mimuw.edu.pl}.} \and
  Micha\l{} Pilipczuk\thanks{Institute of Informatics, University of Warsaw, Poland, \texttt{michal.pilipczuk@mimuw.edu.pl}.}}

  \date{}
\maketitle

\begin{abstract}
We study the complexity of a generic hitting problem \shit{}, where
given a fixed pattern graph $H$ and an input graph $G$, the task is to
find a set $X \subseteq V(G)$ of minimum size that hits all subgraphs
of $G$ isomorphic to $H$.  In the colorful variant of the problem,
each vertex of $G$ is precolored with some color from $V(H)$ and we
require to hit only $H$-subgraphs with matching colors.  Standard
techniques shows that for every fixed $H$, the problem is
fixed-parameter tractable parameterized by the treewidth of $G$;
however, it is not clear how exactly the running time should depend on
treewidth. For the colorful variant, we demonstrate matching upper and
lower bounds showing that the dependence of the running time on
treewidth of $G$ is tightly governed by $\msep(H)$, the maximum size
of a minimal vertex separator in $H$.  That is, we show for every
fixed $H$ that, on a graph of treewidth $t$, the colorful problem can
be solved in time $2^{\Oh(t^{\msep(H)})}\cdot|V(G)|$, but cannot be
solved in time $2^{o(t^{\msep(H)})}\cdot |V(G)|^{O(1)}$, assuming the
Exponential Time Hypothesis (ETH).  Furthermore, we give some
preliminary results showing that, in the absence of colors, the
parameterized complexity landscape of \shit{} is much richer.

%%% Local Variables: 
%%% mode: plain-tex
%%% TeX-master: "tw-subgraph-hitting"
%%% End: 

\end{abstract}

\section{Introduction}\label{sec:intro}

%Obtaining matching upper and lower bounds for
%parameterized complexity of various problems has recently
%became a very active research direction.
%In the recent survey~\cite{marx:survey}, this so-called ``optimality programme'' has been classified to be 
%one of the main future goals
%in parameterized complexity.

The ``optimality programme'' is a thriving trend within parameterized
complexity, which focuses on pursuing tight bounds on the time
complexity of parameterized problems. Instead of just determining
whether the problem is fixed-parameter tractable, that is, whether the
problem with a certain parameter $k$ can be solved in time $f(k)\cdot
n^{\Oh(1)}$ for some computable function $f(k)$, the goal is to
determine the best possible dependence $f(k)$ on the parameter $k$.
For several problems, matching upper and lower bounds have been
obtained on the function $f(k)$. The lower bounds are under the
complexity assumption Exponential Time Hypothesis (ETH), which roughly
states than $n$-variable 3SAT cannot be solved in time $2^{o(n)}$;
see, e.g., the survey of Lokshtanov et al.~\cite{lms:survey}.

One area where this line of research was particularly successful is the
study of fixed-parameter algorithms parameterized by the treewidth of
the input graph and understanding how the running time has to depend
on the treewidth.  Classic results on model checking monadic
second-order logic on graphs of bounded treewidth, such as Courcelle's
Theorem, provide a unified and generic way of proving fixed-parameter
tractability of most of the tractable cases of this
parameterization~\cite{ArnborgLS91,courcelle}. While these results
show that certain problems are solvable in time $f(t)\cdot n$ on graphs of
treewidth $t$ for some function $f$, the exact function $f(t)$ resulting from this approach is usually hard to determine and far from
optimal. To get reasonable upper bounds on $f(t)$, one typically
resorts to constructing a dynamic programming algorithm, which often is straightforward, but tedious.

The question whether the straightforward dynamic programming
algorithms for bounded treewidth graphs are optimal received
particular attention in 2011.  On the hardness side, Lokshtanov, Marx
and Saurabh proved that many natural algorithms are probably
optimal~\cite{lms:known,lms:slightly}. In particular, they showed that
there are problems for which the $2^{\Oh(t\log t)} n$ time algorithms
are best possible, assuming ETH.
%in the sense that, assuming ETH, here are no
%$2^{o(t\log t)}\cdot n^{\Oh(1)}$ time algorithms.
On the algorithmic
side, Cygan et al.~\cite{cut-and-count} presented a new technique,
called {\em{Cut\&Count}}, that improved the running time
of the previously known (natural) algorithms for many connectivity
problems. For example, previously only $2^{\Oh(t\log t)}\cdot n^{\Oh(1)}$
algorithms were known for \textsc{Hamiltonian Cycle} and
\textsc{Feedback Vertex Set}, which was improved to $2^{\Oh(t)}\cdot
n^{\Oh(1)}$ by Cut\&Count.  These results indicated that not only
proving tight bounds for algorithms on tree decompositions is within
our reach, but such a research may lead to surprising algorithmic
developments.  Further work includes derandomization
of Cut\&Count in~\cite{cut-and-count-derand1,FominLS14}, an attempt to
provide a meta-theorem to describe problems solvable in
single-exponential time~\cite{cut-and-count-logic}, and a new
algorithm for \textsc{Planarization}~\cite{planarization}.

We continue here this line of research by investigating a family of
subgraph-hitting problems parameterized by treewidth and find
surprisingly tight bounds for a number of problems. An interesting
conceptual message of our results is that, for every integer
$c\ge 1$, there are fairly natural problems where the best possible
dependence on treewidth is of the form $2^{\Oh(t^c)}$.

% The lower bounds in all aforementioned papers are based
% on the \emph{Exponential Time Hypothesis} (ETH) or
% \emph{Strong Exponential Time Hypothesis} (SETH)~\cite{IP01}.
% Informally speaking (and surpressing some technical details),
% ETH states that no subexponential, in terms
% of the number of variables, algorithm exists for the {\sc{$3$-SAT}} problem, while
% SETH states that no exponential speed-up over brute-force algorithm
% for the general {\sc{CNF-SAT}}.
% For more on ETH- and SETH-based lower bounds we refer to the recent survey~\cite{lms:survey}.

\paragraph{Studied problems and motivation}
In our paper we focus on the following generic \shit{} problem: for a pattern graph $H$
and an input graph $G$, what is the minimum size of a set $X \subseteq V(G)$ that hits
all subgraphs of $G$ that are isomorphic to $H$?
(Henceforth we call them \emph{$H$-subgraphs} for brevity.)
This problem generalizes a few other problems studied in the literature,
for example \textsc{Vertex Cover} (for $H = P_2$)~\cite{lms:known},
or finding the largest induced subgraph
of maximum degree at most $\Delta$ (for $H = K_{1,\Delta+1}$)~\cite{max-deg-vd}. We also study the following \emph{colorful} variant \cshit{}, where the input graph $G$
is additionally equipped with a coloring $\col : V(G) \to V(H)$, and we are only interested
in hitting $H$-subgraphs where every vertex matches its color.

A direct source of motivation for our study is the work of Pilipczuk~\cite{cut-and-count-logic}, which attempted to describe graph problems admitting fixed-parameter algorithms with running time of the form $2^{\Oh(t)}\cdot |V(G)|^{\Oh(1)}$, where $t$ is the treewidth of $G$.
The proposed description is a logical formalism where one can quantify existence of some vertex/edge sets,
 whose properties can be verified ``locally'' by requesting satisfaction of a formula of modal logic in every vertex.
 In particular, Pilipczuk argued that the language for expressing local properties needs to be somehow modal,
 as it cannot be able to discover cycles in a constant-radius neighborhood of a vertex.
This claim was supported by a lower bound: unless ETH fails, for any constant $\ell\ge 5$, the problem of finding the minimum size of a set that hits all the cycles $C_\ell$ in a graph of treewidth $t$ cannot be solved in time $2^{o(t^2)}\cdot |V(G)|^{\Oh(1)}$. Motivated by this result, we think that it is natural to investigate the complexity of hitting subgraphs for more general patterns $H$, instead of just cycles.

%The source of our interest in the colorful variant is twofold.
We may see the colorful variant as an intermediate step towards full
understanding of the complexity of \shit{}, but it is also an
interesting problem on its own.  It often turns out that the
colorful variants of problems are easier to investigate, while their
study reveals useful insights; a remarkable example is the
kernelization lower bound for \textsc{Set Cover} and related
problems~\cite{colors-and-ids}.  In our case, if we allow colors, a
major combinatorial difficulty vanishes: when the algorithm keeps
track of different parts of the pattern $H$ that appear in the graph
$G$, and combines a few parts into a larger one, the coloring $\col$
ensures that the parts are vertex-disjoint.  Hence, 
%our initial
%belief, later supported by the results, was that 
the colorful variant is easier to study, whereas at the same time it
reveals interesting insight into the standard variant.

% The second motivation is that \cshit{} models some situation appearing in Constraint Satisfaction Programming.
% Assume we are given a binary CSP instance $(V,\Sigma,(R_{u,v})_{uv \in E})$, where
% $V$ is the set of variables, $\Sigma$ is the domain, and $R_{u,v} \subseteq \Sigma \times \Sigma$
% is some binary relation for $uv \in E \subseteq \binom{X}{2}$. Our goal is to find
% an assignment $\phi: V \to \Sigma$ such that $(\phi(u),\phi(v)) \in R_{u,v}$ for each $uv \in E$.
% Construct a graph $G$ where $V(G) = V \times \Sigma$ and $(u,\alpha)(v,\beta) \in E(G)$ if and only if
% $(\alpha,\beta) \in R_{u,v}$ (i.e., we match two variable-value pairs if fixing them in an assignment
% would satisfy the constraint between the variables in question).
% Moreover, define a ``solution'' pattern $H = (V,E)$ and a coloring $\col(v,\alpha) = v$.
% Observe that a feasible solution to \cshit{} on $(G,\col)$ corresponds to a set that hits
% all solutions to the original CSP instance.\todo{citations?}

\paragraph{Our results and techniques}
% Our study reveals two parameters of the pattern graph $H$ that are important for the problems in question.
% The first one, denoted $\msep(H)$, is the maximum size of a minimal (vertex) separator in $H$.
% The second one, denoted $\mnei(H)$, is the maximum size of $N_H(A)$, where $A$ iterates
% over connected subsets of $V(H)$ such that $N_H(N_H[A]) \neq \emptyset$, i.e.,
% $N_H[A]$ is not a whole connected component of $H$.
% Observe that $\msep(H) \leq \mnei(H)$ for any $H$.
In the case of \cshit{}, we obtain a tight bounds for the complexity
of the treewidth parameterization.  First, note that, in the presence
of colors, one can actually solve \cshit{} for each connected
component of $H$ independently; hence, we may focus only on connected
patterns $H$.  Second, we observe that there are two special cases. If
$H$ is a path then \cshit{} reduces to a maximum flow/minimum cut
problem, and hence is polynomial-time solvable.  If $H$ is a clique,
then any $H$-subgraph of $G$ needs to be contained in a single bag of
any tree decomposition, and there is a simple $2^{\Oh(t)} |V(G)|$-time
algorithm, where $t$ is the treewidth of $G$.
Finally, for the remaining cases we show that the dependence on
treewidth is tightly connected to the value of $\msep(H)$, the maximum
size of a minimal vertex separator in $H$ (a separator $S$ is minimal
if there are two vertices $x,y$ such that $S$ is an $xy$-separator,
but no proper subset of $S$ is). We prove the following matching upper
and lower bounds.

\begin{theorem}\label{thm:intro:cshit:algo}
A \cshit{} instance $(G,\col)$
can be solved in time $2^{\Oh(t^{\msep(H)})} |V(G)|$
in the case when $H$ is connected and is not a clique,
   where $t$ is the treewidth of $G$.
\end{theorem}

\begin{theorem}\label{thm:intro:lb:col}
Let $H$ be a graph that contains a connected component that is neither a path nor a clique.
Then, unless ETH fails, there does
not exist an algorithm that, given a \cshit{} instance $(G,\col)$
and a tree decomposition of $G$ of width $t$, resolves $(G,\col)$
in time $2^{o(t^{\msep(H)})} |V(G)|^{\Oh(1)}$.
\end{theorem}

In every theorem of this paper, we treat $H$ as a fixed graph of constant size, and hence the factors hidden in the $\Oh$-notation may depend on the size of $H$.

In the absence of colors, we give preliminary results showing that the parameterized
complexity of the treewidth parameterization of \shit{} is more involved than
the one of the colorful counterpart.
In this setting, we are able to relate the dependence on treewidth only to a larger parameter of the graph $H$. 
Let $\mnei(H)$ be the maximum size of $N_H(A)$, where $A$ iterates
over connected subsets of $V(H)$ such that $N_H(N_H[A]) \neq \emptyset$, i.e.,
$N_H[A]$ is not a whole connected component of $H$.
Observe that $\msep(H) \leq \mnei(H)$ for any $H$.
First, we were able to construct a counterpart of Theorem~\ref{thm:intro:cshit:algo}
only with the exponent $\mnei(H)$.

\begin{theorem}\label{thm:std:algo}
Assume $H$ contains a connected component that is not a clique.
Then, given a graph $G$ of treewidth $t$,
one can solve \shit{} on $G$ in time $2^{\Oh(t^{\mnei(H)} \log t)} |V(G)|$. 
\end{theorem}

 We remark that  for \cshit{}, an algorithm with running time
 $2^{\Oh(t^{\mnei(H)})} |V(G)|$
 (as opposed to $\msep(H)$ in the exponent in Theorem~\ref{thm:intro:cshit:algo})
 is rather straightforward: in the state of
 dynamic programming one needs to remember, for every subset $X$ of the bag of size at most $\mnei(G)$, all forgotten connected parts of $H$ that are attached to $X$ and not hit
 by the constructed solution. To decrease the exponent to $\msep(H)$, we introduce a
 ``prediction-like'' definition of a state of the dynamic programming,
 leading to highly involved proof of correctness.
For the problem without colors, however, even an algorithm with the exponent $\mnei(H)$ (Theorem~\ref{thm:std:algo})
is far from trivial. We cannot limit ourselves to keeping track of forgotten
connected parts of the graph $H$ independently of each other, since in the absence of colors
these parts may not be vertex-disjoint and, hence, we would not be able to reason
about their union in latter bags of the tree decomposition.
To cope with this issue, we show that the set of forgotten
(not necessarily connected) parts of the graph $H$ that are subgraphs of $G$
can be represented as a \emph{witness graph} with $\Oh(t^{\mnei(H)})$ vertices and edges.
As there are only $2^{\Oh(t^{\mnei(H)} \log t)}$ possible graphs of this size,
the running time bound follows.

We also observe that the bound of $\Oh(t^{\mnei(H)})$ on the size of the witness graph
is not tight for many patterns $H$. For example, if $H$ is a path, then
we are able to find a witness graph with $\Oh(t)$ vertices and edges, and the algorithm of Theorem~\ref{thm:std:algo}
runs in time $2^{\Oh(t \log t)} |V(G)|$.

From the lower bound perspective, we were not able to prove an analogue of
Theorem~\ref{thm:intro:lb:col} in the absence of colors. However, there is a good reason
for that: we show that for any fixed $h \geq 2$ and $H = K_{2,h}$,
the \shit{} problem is solvable in time $2^{\Oh(t^2 \log t)} |V(G)|$
for a graph $G$ of treewidth $t$.
This should be put in contrast with $\mnei(K_{2,h}) = \msep(K_{2,h}) = h$.
Moreover, the lower bound of $2^{o(t^h)}$ can be proven if we break the symmetry
of $K_{2,h}$ by attaching a triangle to each of the two degree-$h$ vertices of $K_{2,h}$ (obtaining a graph
we denote by $H_h$; see Figure~\ref{fig:Hh}).
\begin{theorem}\label{thm:lb:Hh}
Unless ETH fails, for every $h \geq 2$ there does
not exist an algorithm that, given a \shitarg{H_h} instance $G$
and a tree decomposition of $G$ of width $t$, resolves $G$
in time $2^{o(t^h)} |V(G)|^{\Oh(1)}$.
\end{theorem}
This indicates that the optimal dependency on $t$ in an algorithm
for \shit{} may heavily rely on the symmetries of $H$, and may 
be more difficult to pinpoint.

%An interesting side result of the lower bound for $H_h$ is that it does not need to assume $h$ is a constant,
%but it can be given as an input to the algorithm. This yields the following interesting double-exponential lower bound. 
%\begin{corollary}\label{cor:lb:Hh:double-exp}
%Unless ETH fails, there does not exist an algorithm that, given a graph $G$ with a tree decomposition of width $t$,
%and an integer $h = \Oh(\log |V(G)|)$, finds in $2^{2^{o(t)}} |V(G)|^{\Oh(1)}$ time
%the minimum size of a set that hits all $H_h$-subgraphs of $G$.
%\end{corollary}

\paragraph{Organization of the paper}
After setting notation in Section~\ref{sec:prelims},
we prove Theorem~\ref{thm:std:algo} in Section~\ref{sec:std:algo},
with a special emphasize on the existence of the witness graph in the begining of the section.
We discuss the special cases of \shit{} in Section~\ref{sec:std-discussion}.
The proofs of results for the colorful variant,
namely Theorems~\ref{thm:intro:cshit:algo} and~\ref{thm:intro:lb:col},
are provided in Sections~\ref{sec:algo-col} and~\ref{sec:lb-col}, respectively.
Section~\ref{sec:conc} concludes the paper.

\section{Preliminaries}\label{sec:prelims}

\paragraph{Graph notation}
In most cases, we use standard graph notation.
A graph $P_n$ is a path on $n$ vertices, a graph $K_n$ is a complete graph on $n$ vertices, and a graph $K_{a,b}$ is a complete bipartite graph with $a$ vertices on one side, and $b$ vertices on the other side.
A \emph{$t$-boundaried graph} is a graph $G$ with a prescribed (possibly empty) \emph{boundary}
$\bd G \subseteq V(G)$ with $|\bd G|\leq t$, and an injective function
$\bdfun_G: \bd G \to \{1,2,\ldots,t\}$. For a vertex $v \in \bd G$
the value $\bdfun_G(v)$ is called the \emph{label} of $v$.

A \emph{colored graph} is a graph $G$ with a function $\col:V(G) \to \labset$,
where $\labset$ is some finite set of colors.
A graph $G$ is $H$-colored, for some other graph $H$, if $\labset = V(H)$.
We also say in this case that $\col$ is an $H$-coloring of $G$.

A \emph{homomorphism} from graph $H$ to graph $G$ is a function $\pi : V(H) \to V(G)$
such that $ab \in E(H)$ implies $\pi(a)\pi(b) \in E(G)$.
In the $H$-colored setting, i.e., when $G$ is $H$-colored, we also require that $\col(\pi(a)) = a$ for any $a \in V(H)$
(every vertex of $H$ is mapped onto appropriate color).
The notion extends also to $t$-boundaried graphs:
if both $H$ and $G$ are $t$-boundaried, we require that
whenever $a \in \bd H$ then $\pi(a) \in \bd G$ and $\bdfun_G(\pi(a)) = \bdfun_H(a)$. Note, however, that we allow that a vertex of $V(H) \setminus \bd H$ is mapped onto a vertex of $\bd G$.
%A homomorphism of $t$-boundaried graphs is \emph{boundary-preserving}
%if additionally $\pi(a) \notin \bd H$ implies $a \notin \bd G$.

An \emph{$H$-subgraph of $G$} is any injective homomorphism $\pi: V(H) \to V(G)$.
Recall that in the $t$-boundaried setting, we require that the labels are preserved,
whereas in the colored setting, we require that the homomorphism respects colors.
In the latter case, we call it a \emph{$\col$-$H$-subgraph of $G$} for clarity.

We say that a set $X \subseteq V(G)$ \emph{hits} a ($\col$-)$H$-subgraph $\pi$
if $X \cap \pi(V(H)) \neq \emptyset$.
The (\textsc{Colorful}) \shit{} problem asks for a minimum possible size
of a set that hits all ($\col$-)$H$-subgraphs of $G$.

\paragraph{(Nice) tree decompositions}
For any nodes $w,w'$ in a rooted tree $\Tree$, we say that $w'$ is a \emph{descendant} of $w$
(denoted $w' \preceq w$) if $w$ lies on the unique path between $w'$ and $\texttt{root}(\Tree)$,
the root of $\Tree$.
A \emph{tree decomposition} of a graph is a pair $(\Tree,\Tbag)$, where
$\Tree$ is a rooted tree, and $\Tbag : V(\Tree) \to 2^{V(G)}$ is a mapping satisfying:
\begin{itemize}
  \item for each vertex $v \in V(G)$, the set $\{w \in V(\Tree)\ |\ v \in \Tbag(w)\}$ induces a nonempty and connected subtree of~$\Tree$,
  \item for each edge $e \in E(G)$, there exists $w \in V(\Tree)$ such that $e \subseteq \Tbag(w)$.
\end{itemize}
The width of a decomposition $(\Tree,\Tbag)$ equals $\max_{w \in V(\Tree)} |\Tbag(w)|-1$,
and the treewidth of a graph is the minimum possible width of its decomposition.

For a tree decomposition $(\Tree,\Tbag)$, we define two auxiliary mappings:
\begin{align*}
\Tall(w) &= \bigcup_{w' \preceq w} \Tbag(w'), & \Tdown(w) &= \Tall(w) \setminus \Tbag(w).
\end{align*}

In our dynamic programming algorithms, it is convenient to work on the so-called
\emph{nice tree decompositions}. A tree decomposition $(\Tree,\Tbag)$ is \emph{nice} if 
$\Tbag(\texttt{root}(\Tree)) = \emptyset$ and each node $w \in V(\Tree)$ is of one of the following four
types:
\begin{description}
\item[leaf node] $w$ is a leaf of $\Tree$ and $\Tbag(w) = \emptyset$.
\item[introduce node] $w$ has exactly one child $w'$, and there exists a vertex $v \in V(G) \setminus \Tbag(w')$ such that $\Tbag(w) = \Tbag(w') \cup \{v\}$.
\item[forget node] $w$ has exactly one child $w'$, and there exists a vertex $v \in \Tbag(w')$,
  such that $\Tbag(w) = \Tbag(w') \setminus \{v\}$.
\item[join node] $w$ has exactly two children $w_1,w_2$ and $\Tbag(w) = \Tbag(w_1) = \Tbag(w_2)$.
\end{description}
It is well known (see e.g.~\cite{nice-decomp})
that any tree decomposition of width $t$ can be transformed, without increasing its width,
into a nice decomposition with $\Oh(t|V(G)|)$ nodes.

Hence, by an application of the recent $5$-approximation for treewidth~\cite{tw-apx}, in all our algorithmic results we implicitely assume that we are given
a nice tree decomposition of $G$ with $\Oh(t|V(G)|)$ nodes and of width \emph{less} than $t$,
  so that each bag is of size at most $t$.
(This shift of the value of $t$ by one is irrelevant for the complexity bounds,
 but makes the notation much cleaner.)
Moreover, we may assume that we also have a function $\labfun: V(G) \to \{1,2,\ldots,t\}$
such that, for each node $w \in V(\Tree)$, $\labfun|_{\Tbag(w)}$ is injective.
(Observe that it is straightforwad to construct $\labfun$ in a top-bottom manner.
) Consequently, we may treat each graph $G[\Tall(w)]$ as a $t$-boundaried
graph with $\bd G[\Tall(w)] = \Tbag(w)$ and labeling $\labfun|_{\Tbag(w)}$.

\begin{figure}[t]
\centering
\includegraphics{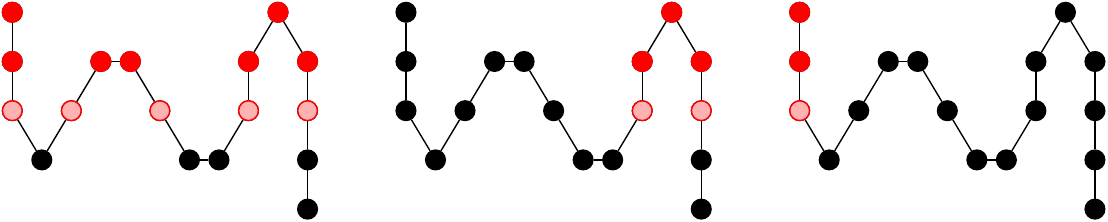}
\caption{Red vertices denote a slice (left), chunk (centre) and separator chunk (right) in a graph $H$ being a path.
The light-red vertices belong to the boundary.}
\label{fig:params}
\vskip -0.4cm
\end{figure}

\begin{figure}[tb]
\centering
\begin{subfigure}{.4\textwidth}
\centering
\includegraphics{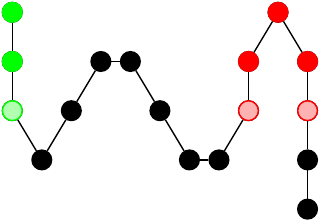}
\caption{In a path $P_h$, $h\geq 5$, we have $\msep(P_h) = 1$ (a separator chunk is depicted green)
whereas $\mnei(P_h) = 2$ (a corresponding chunk is depicted red).}
\end{subfigure} \quad
\begin{subfigure}{.4\textwidth}
 \centering
 \includegraphics{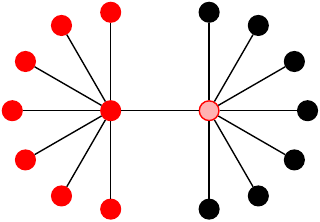}
 \caption{In a double star we have $\msep(H) = \mnei(H) = 1$ (a chunk is depicted red).}
\end{subfigure}\\\vspace{2mm}%
\begin{subfigure}{.4\textwidth}
 \centering
 \includegraphics{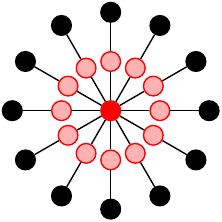}
 \caption{In a subdivided star $\mnei(H)$ equals the degree of the center (a corresponding chunk is depicted red), whereas
   $\msep(H) = 1$ as in all trees on at least three vertices.}
\end{subfigure}\quad
\begin{subfigure}{.4\textwidth}
 \centering
 \includegraphics{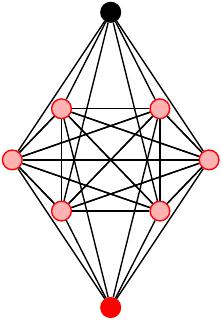}
 \caption{In a clique without one edge, $\msep(H) = \mnei(H) = |V(H)|-2$ (a corresponding separator chunk is depicted red).}
\end{subfigure}\\\vspace{2mm}%
\begin{subfigure}{.4\textwidth}
 \centering
 \includegraphics{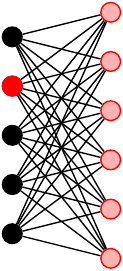}
 \caption{In a biclique $K_{a,b}$, $a,b \ge 2$, we have $\msep(H) = \mnei(H) = \max(a,b)$ (a corresponding separator chunk is depicted red).}
\end{subfigure}
\caption{Examples of graphs with the values of $\msep(H)$ and $\mnei(H)$.
  In each example, the vertices with lighter color belong to the boundary of a corresponding chunk.
\label{fig:examples}}
\end{figure}
    
\paragraph{Important graph invariants, chunks, and slices}
For two vertices $a,b \in V(H)$, a set $S \subseteq V(H) \setminus \{a,b\}$
is an \emph{$ab$-separator} if $a$ and $b$ are not in the 
same connected component of $H \setminus S$.
The set $S$ is additionally a \emph{minimal $ab$-separator}
if no proper subset of $S$ is an $ab$-separator.
A set $S$ is a \emph{minimal separator} if it is a minimal $ab$-separator
for some $a,b \in V(H)$.
For a graph $H$, by $\msep(H)$ we denote the maximum size
of a minimal separator in $H$.

For an induced subgraph $H' = H[D]$, $D \subseteq V(H)$,
we define the boundary $\bd H' = N_H(V(H) \setminus D)$ and the interior $\inte H' = D \setminus \bd H[D]$; thus $V(H')=\bd H'\uplus \inte H'$.
Observe that $N_H(\inte H') \subseteq \bd H'$; the inclusion can be proper, as it is possible that a vertex of $\bd H'$ has no neighbor in $\inte H'$.
An induced subgraph $H'$ of $H$ is a \emph{slice} if $N_H(\inte H') = \bd H'$,
and a \emph{chunk} if additionally $H[\inte H']$ is connected.
For a set $A \subseteq V(H)$, we use $\slice[A]$ ($\chunk[A]$) to denote the unique slice (chunk) with interior $A$ (if it exists).
The intuition behind this definition is that, when we consider some bag $\Tbag(w)$
in a tree decomposition, a slice is a part of $H$ that may already be present in $G[\Tall(w)]$
and we want to keep track of it.
If a slice (chunk) $\slice$ is additionally equipped with an injective labeling $\bdfun_\slice : \bd \slice \to \{1,2,\ldots,t\}$,
then we call the resulting $t$-boundaried graph a \emph{$t$-slice} (\emph{$t$-chunk}, respectively).

By $\mnei(H)$ we denote the maximum size of $\bd \chunk$, where $\chunk$ iterates over all chunks of $H$.
We remark here that both $\msep(H)$ and $\mnei(H)$
are positive only for graphs $H$ that contain at least one connected
component that is not a clique, as otherwise there are no chunks with nonempty boundary nor minimal separators in $H$.

Observe that if $S$ is a minimal $ab$-separator in $H$,
and $A$ is the connected component of $H \setminus S$ that
contains $a$, then $N_H(A) = S$ and $\chunk[A]$ is a
chunk in $H$ with boundary $S$. Consequently,
$\msep(H) \leq \mnei(H)$ for any graph $H$.
A chunk $\chunk$ for which $\bd \chunk$ is a minimal separator in $H$
is henceforth called a \emph{separator chunk}.
Equivalently, a chunk $\chunk$ is a separator chunk if and only if there exists
a connected component $B$ of $H \setminus \chunk$ such that $N_H(B) = \bd \chunk$.
Observe also that $\msep(H)$ equals the maximum boundary size over all separator chunks in $H$.

See also Figures~\ref{fig:params} and~\ref{fig:examples} for an illustration.

\paragraph{Exponential Time Hypothesis and SAT instances}
Our lower bounds are based on the \emph{Exponential Time Hypothesis} (ETH)
of Impagliazzo and Paturi~\cite{IP01}. 
We do not need here its formal definition, but instead we rely on the following corollary
of the celebrated Sparsification Lemma~\cite{IPZ01}.
\begin{theorem}[\cite{IPZ01}]\label{thm:spars}
Unless ETH fails, there does not exist an algorithm that can resolve satisfiability of a
$n$-variable $m$-clause $3$-CNF formula in time $2^{o(n+m)}$.
\end{theorem}

In our reductions we start from a slightly preprocessed $3$-SAT instances.
We say that a $3$-CNF formula $\Phi$ is \emph{clean} if
each variable of $\Phi$ appears exactly three times, at least once
positively and at least once negatively, and each clause of $\Phi$ contains
two or three literals and does not contain twice the same variable.
\begin{lemma}\label{lem:preprocess}
Given any $3$-CNF formula $\Phi$, one can in polynomial time compute
an equivalent clean formula $\Phi'$ of size linearly bounded in the size of $\Phi$, such that $\Phi'$ is satisfiable if and only if $\Phi$ is.
\end{lemma}
\begin{proof}
First, preprocess $\Phi$ as follows:
\begin{enumerate}
\item Simplify all clauses that contain repeated variables: delete all clauses
that contain both a variable and its negation, and remove duplicate literals
from the remaining clauses.
\item As long as there exists a variable $x$ that appears only positively or only
negatively, fix the evaluation of $x$ that satisfies all clauses containing $x$
and simplify the formula.
\item As long as there exists a clause $C$ with only one literal, fix the evaluation
of this literal that satisfies $C$ and simplify the formula.
\item If some clause becomes empty in the process, return a dummy unsatisfiable clean formula.
\end{enumerate}
Observe that after this preprocessing the size of $\Phi$ could have only shrunk,
while each variable now appears at least twice.

Second, replace every variable $x$ with a cycle of implications
$x_1 \Rightarrow x_2 \Rightarrow \ldots \Rightarrow x_{s(x)} \Rightarrow x_1$, where
$s(x)$ is the number of appearances of $x$ in $\Phi$.
More formally, for every variable $x$:
\begin{enumerate}
\item introduce $s(x)$ new variables $x_1, x_2, \ldots, x_{s(x)}$, and replace
each occurence of $x$ in $\Phi$ with a distinct variable $x_i$; and
\item introduce $s(x)$ new clauses $x_i \Rightarrow x_{i+1}$ (i.e., $\neg x_i \vee x_{i+1}$)
  for $i=1,2,\ldots,s(x)$, where $x_{s(x)+1} = x_1$.
\end{enumerate}
Observe that, after this replacement, each variable $x_i$ appears exactly three times in
the formula $\Phi$: positively in the implication $x_{i-1} \Rightarrow x_i$,
negatively in the implication $x_i \Rightarrow x_{i+1}$ (with the convention $x_0 = x_{s(x)}$
and $x_{s(x)+1} = x_1$), and the third time in one former literal of the variable $x$.
Moreover, as after the first step each variable appears at least twice in $\Phi$,
for every former
variable $x$ we have $s(x) \geq 2$ and no new clause contains twice the same variable.
Finally, note that the second step increases the size of the formula only by a constant factor.
The lemma follows.
\maybeqed\end{proof}

\section{General algorithm for \shit{}}\label{sec:std:algo}

In this section we present an algorithm for \shit{} running in time
$2^{\Oh(t^{\mnei(H)} \log t)} |V(G)|$,
where $t$ is the width of the tree decomposition we are working on.
The general idea is the natural one.
\begin{definition}[profile]
The \emph{profile} of a $t$-boundaried graph $(G,\bdfun)$ is the set $\slicefam^{G,\bdfun}$
of all $t$-slices that are subgraphs of $(G,\bdfun)$.
A \emph{$(\leq p)$-profile} of $(G,\bdfun)$ is the profile of $(G,\bdfun)$, restricted to all $t$-slices that have at most $p$ vertices.
\end{definition}
For each node $w$ of the tree decomposition,
for each set $\wX \subseteq \Tbag(w)$,
and for each family $\slicefam$ of $t$-slices, we would like to find the minimum size
of a set $X \subseteq \Tdown(w)$ such that, if we treat $G[\Tall(w)]$ as a $t$-boundaried
graph with $\bd G[\Tall(w)] = \Tbag(w)$ and labeling $\labfun|_{\Tbag(w)}$, then
the profile of  $G[\Tall(w) \setminus (X \cup \wX)]$ is exactly $\slicefam$.
However, as the number of $t$-slices can be as many as $t^{|H|}$, we have too many choices for the potential profile $\slicefam$.

\subsection{The witness graph}

The essence of the proof
is to show that each ``reasonable'' choice of $\slicefam$
can be encoded as a \emph{witness graph} of essentially size $\Oh(t^{\mnei(H)})$.
Such a claim would give a $2^{\Oh(t^{\mnei(H)} \log t)}$ bound on the number of possible
witness graphs, and provide a good bound on the size of state space.

For technical reasons, we need to slightly generalize the notion of a profile, so that
it encapsulates also the possibility of ``reserving'' a small set $Y$ of vertices from the boundary for another pieces of the graph $H$.
\begin{definition}[extended profile]
For a $t$-boundaried graph $(G,\bdfun)$, \emph{an extended profile} 
consists of a $(\leq |V(H)|-|Y|)$-profile $\slicefam^{G,\bdfun}_Y$ of the graph $(G \setminus Y,\bdfun|_{\bd G \setminus Y})$ for every set $Y \subseteq \bd G$ of size at most $|V(H)|$.
\end{definition}
\begin{definition}[witness-subgraph, equivalent]
Let $(G_1,\bdfun_1)$ and $(G_2,\bdfun_2)$ be two $t$-boundaried graphs. We say that $(G_1,\bdfun_1)$ is a \emph{witness-subgraph} of $(G_2,\bdfun_2)$, 
if $\bd G_1 = \bd G_2$, $\bdfun_1 = \bdfun_2$, and, moreover, for every $Y \subseteq \bd G_1$ of size at most $|V(H)|$ we have
$\slicefam^{G_1,\bdfun_1}_Y \subseteq \slicefam^{G_2,\bdfun_2}_Y$.
We say that two $t$-boundaried graphs are \emph{equivalent}, if one is the witness-subgraph of the other, and vice-versa (i.e., they boundaries, labelings, and extended profiles are equal).
\end{definition}
Let us emphasize that, maybe slightly counterintuitively, a witness-subgraph is not necessarily a subgraph; the `sub' term corresponds to admitting a subfamily of $t$-slices as subgraphs.

In the next lemma we show that every $t$-boundaried graph has a small equivalent one, showing us that there is only a small number
of reasonable choices for an (extended) profile in the dynamic programming algorithm.
\begin{lemma}\label{lem:witness}
Assume $H$ contains a connected component that is not a clique.
Then, for any $t$-boundaried graph $(G,\bdfun)$ there exists an equivalent $t$-boundaried graph $(\wG,\bdfun)$ such
that (a) is a subgraph of $(G,\bdfun)$, (b) $\bd G = \bd \wG$ and $G[\bd G] = \wG[\bd \wG]$, and (c)
$\wG \setminus E(\wG[\bd \wG])$ contains $\Oh(t^{\mnei(H)})$ vertices and edges.
%  for any $t$-slice
%$\slice$ and any set $Y \subseteq V(G)$ such that $|Y| + |V(\slice)| \leq |V(H)|$,
%there exists a $\slice$-subgraph in $(G \setminus Y,\bdfun)$ if and only
%if there exists one in $(\wG \setminus Y,\bdfun)$.
%%% Do not need the algorithmic part :-)
%Moreover, the graph $(\wG,\bdfun)$ can be constructed in time
%$\Oh(t^{\Oh(|V(H)|)} |V(G)|)$ if a tree decomposition of $G$ width $t$ is provided.
\end{lemma}

\begin{proof}
We define $\wG$ by a recursive procedure. We start with $\wG = G[\bd G]$.
Then, for every $t$-chunk $\chunk = (H',\bdfun')$,
we invoke a procedure
$\operadd(\chunk,\emptyset)$. The procedure $\operadd(\chunk,X)$, for $X \subseteq V(G)$,
first tries to find a $\chunk$-subgraph $\pi$ in
$(G \setminus X,\bdfun)$. If there is none, the procedure terminates.
Otherwise, it first adds all edges and vertices of $\pi(\chunk)$ to $\wG$ that
are not yet present there.
Second, if $|X| < |V(H)|$, then it recursively invokes $\operadd(\chunk,  X \cup \{v\})$
for each $v \in \pi(\chunk)$.

We first bound the size of the constructed graph $\wG$.
There are at most $2^{|V(H)|} t^{\mnei(H)}$ choices for the chunk,
since a chunk $\chunk$ is defined by its vertex set, and there are at most $t^{\mnei(H)}$ labelings
of its boundary.
The procedure $\operadd(\chunk, X)$ at each step adds at most one copy of $H$ to $G$,
and branches into at most $|V(H)|$ directions. The depth of the recursion is bounded
by $|V(H)|$. Hence, in total at most $2^{|V(H)|}t^{\mnei(H)}\cdot (|V(H)|+|E(H)|)\cdot |V(H)|^{|V(H)|}=\Oh(t^{\mnei(H)})$ edges and vertices are added to $\wG$,
except for the initial graph $G[\bd G]$ (recall that we consider $H$ to be a fixed graph and the constants hidden by the big-$\Oh$ notation can depend on $H$).
%Moreover, each call to $\operadd(X)$ attempts to find one boundary-preserving
%$\chunk$-subgraph of $G \setminus X$.
%Using the tree decomposition of $G$ of width $t$, a straightforward dynamic programming
%algorithm performs this task in time $t^{\Oh(|V(H)|)} |V(G)|$.
%The running time bound follows.

It remains to argue that $\wG$ satisfies property (d). Clearly, since $(\wG,\bdfun)$
is a subgraph of $(G,\bdfun)$, the implication in one direction is trivial.
In the other direction, we start with the following claim.

\begin{myclaim}\label{cl:witness:chunk}
For any set $Z \subseteq V(G)$ of size at most $|V(H)|$, and for any
$t$-chunk $\chunk$,
if there exists a $\chunk$-subgraph
in $(G \setminus Z, \bdfun)$ then there exists also one in $(\wG \setminus Z,\bdfun)$.
\end{myclaim}
\begin{proof}
Let $\pi$ be a $\chunk$-subgraph in $(G \setminus Z,\bdfun)$.
Define $X_0 = \emptyset$. We will construct sets $X_0 \subsetneq X_1 \subsetneq \ldots$,
where $X_i \subseteq Z$ for every $i$, and 
analyse the calls to the procedure $\operadd(\chunk, X_i)$ in the process of constructing $\wG$.

Assume that $\operadd(\chunk, X_i)$ has been invoked at some point during the construction;
clearly this is true for $X_0 = \emptyset$. Since we assume $X_i \subseteq Z$,
there exists a
$\chunk$-subgraph in $(G \setminus X_i,\bdfun)$ --- $\pi$ is one such example.
Hence, $\operadd(\chunk, X_i)$ has found a $\chunk$-subgraph $\pi_i$, and added its image to $\wG$.
If $\pi_i$ is a $\chunk$-subgraph also in $(\wG \setminus Z,\bdfun)$, then we are done.
Otherwise, there exists $v_i \in Z \setminus X_i$ that is also present in the image of $\pi_i$.
In particular, since $|Z| \leq |V(H)|$, we have $|X_i| < |V(H)|$ and the call
$\operadd(\chunk, X_i \cup \{v_i\})$ has been invoked. We define $X_{i+1} := X_i \cup \{v_i\}$.

Since the sizes of sets $X_i$ grow at each step, for some $X_i$, $i \leq |Z|$, we reach
the conclusion that $\pi_i$ is a $\chunk$-subgraph of $(\wG \setminus Z,\bdfun)$, and the claim is proven.
\cqed\end{proof}

Fix now a set $Y \subseteq V(G)$ and a $t$-slice $\slice$ with labeling $\bdfun_\slice$
and with $|Y| + |V(\slice)| \leq |V(H)|$.
Let $\pi$
be a $\slice$-subgraph of $(G \setminus Y,\bdfun)$.
Let $A_1,A_2,\ldots,A_r$ be the connected components of $H[\inte \slice]$.
Define $H_i = N_H[A_i]$, and observe that each $H_i$ is a chunk
with $\bd H_i = N_H(A_i) \subseteq \bd \slice$.
We define $\bdfun_i = \bdfun_\slice|_{\bd H_i}$ to obtain a $t$-chunk $\chunk_i = (H_i,\bdfun_i)$.
By the properties of a $t$-slice, each vertex of $\slice$ is present in at least
one graph $\chunk_i$, and vertices of $\bd \slice$ may be present in more than one.

We now inductively define injective 
homomorphisms $\pi_0,\pi_1,\ldots,\pi_r$ such that
$\pi_i$ maps the subgraph of $\slice$ induced by $\bd \slice \cup \bigcup_{j \leq i} A_j$
to $(\wG \setminus Y,\bdfun)$, and does not use any vertex
of $\bigcup_{j > i}\pi(A_j)$. Observe that $\pi_r$ is
a $\slice$-subgraph of $(\wG \setminus Y,\bdfun)$.
Hence, this construction will conclude the proof of the lemma.

For the base case, recall that $\pi(\bd \slice) \subseteq \bd G = \bd \wG$ and define $\pi_0 = \pi|_{\bd \slice}$.
For the inductive case, assume that
$\pi_{i-1}$ has been constructed for some $1 \leq i \leq r$.
Define 
$$Z_i = Y \cup \pi(\bd \slice \setminus \bd H_i) \cup \bigcup_{j < i}\pi_{i-1}(A_j) \cup \bigcup_{j > i} \pi(A_j).$$
Note that since $\pi$ and $\pi_{i-1}$ are injective and $Y$ is disjoint with $\pi(\bd \slice)$, then we have that $Z_i\cap \pi(\bd \slice) = \pi(\bd \slice \setminus \bd H_i)$.
This observation and the inductive assumption on $\pi_{i-1}$ imply that the mapping $\pi|_{V(H_i)}$ does not use any vertex of $Z_i$.
Thus, $\pi|_{V(H_i)}$ is a $\chunk_i$-subgraph in $(G \setminus Z_i, \bdfun)$.
Observe moreover that $|Z_i| \leq |Y|+|V(\slice)|\leq |V(H)|$.
By Claim~\ref{cl:witness:chunk}, there exists a $\chunk_i$-subgraph $\pi_i'$ in
$(\wG \setminus Z_i,\bdfun)$. Observe that, since $\pi_i'$ and $\pi_{i-1}$
are required to preserve labelings on boundaries of their preimages, 
$\pi_i := \pi_i' \cup \pi_{i-1}$ is a function and a homomorphism.
Moreover, by the definition of $Z_i$, $\pi_i$ is injective
and does not use any vertex of $\bigcup_{j > i}\pi(A_j)$.
Hence, $\pi_i$ satisfies all the required conditions, and the inductive construction is completed.
This concludes the proof of the lemma.
\maybeqed\end{proof}

\subsection{The dynamic programming algorithm}\label{ss:std-algo}

Using Lemma~\ref{lem:witness}, we now define states of the dynamic programming algorithm
on the input tree decomposition $(\Tree,\Tbag)$.
For every node $w \in V(\Tree)$, a \emph{state}
is a pair $\state = (\wX,\wG)$ where $\wX \subseteq \Tbag(w)$
and $\wG$ is a graph with $\Oh(t^{\mnei(H)})$ vertices and edges
such that $\Tbag(w) \setminus \wX \subseteq V(\wG)$ and $\wG[\Tbag(w) \setminus \wX] = G[\Tbag(w) \setminus \wX]$.
We treat $\wG$ as a $t$-boundaried graph with $\bd \wG = \Tbag(w) \setminus \wX$
and labeling $\labfun|_{\Tbag(w) \setminus \wX}$.
We say that
a set $X \subseteq \Tdown(w)$ is \emph{feasible} for $w$ and $\state$ if
$(G[\Tall(w) \setminus (X \cup \wX)], \labfun|_{\Tbag(w) \setminus \wX})$ 
is a witness-subgraph of
$(\wG,\labfun|_{\Tbag(w) \setminus \wX})$.
%$G[\Tall(w)] \setminus (X \cup \wX)$ does not contain an $H$-subgraph and,
%for every $Y \subseteq \Tbag(w) \setminus \wX$ and for every $t$-slice $\slice$
%such that $|Y| + |V(\slice)| \leq |V(H)|$, if there is a
%$\slice$-subgraph in $(G[\Tall(w) \setminus (X \cup \wX \cup Y)], \labfun|_{\Tbag(w) \setminus (\wX\cup Y)})$ then there is also one in
%$(\wG \setminus Y,\labfun|_{\Tbag(w) \setminus (\wX\cup Y)})$.
For every $w$ and every state $\state$, we would like to compute $T[w,\state]$,
the minimum possible size of a feasible set $X$.
Note that by Lemma~\ref{lem:witness} the answer to the input \shit{} instance is the minimum value of $T[\mathtt{root}(\Tree),(\emptyset, \wG)]$
where $\wG$ iterates over all graphs with $\Oh(t^{\mnei(H)})$ vertices and edges that do not contain the $t$-slice $(H,\emptyset)$ as a subgraph.
Hence, it remains to show how to compute the values $T[w,\state]$ in a bottom-up manner in the tree decomposition.

Observe that, if we have two $t$-boundaried graphs $(G_1,\bdfun_1)$ and $(G_2,\bdfun_2)$ of size
$t^{\Oh(1)}$ each (e.g., they are parts of a state, or they were obtained from Lemma~\ref{lem:witness}),
a brute-force algorithm checks the relation of being a witness-subgraph in $t^{\Oh(1)}$ time.

We start the description with the following auxiliary definition.
Let $P \subseteq V(H)$. Observe that $N_H[\inte H[P]] \subseteq P$ and $H[N_H[\inte H[P]]]$
is the unique inclusion-wise maximal slice with vertex set being a subset of $P$.
We call the slice $H[N_H[\inte H[P]]]$ the \emph{core slice} of $P$, and the remaining
vertices $P \setminus N_H[\inte H[P]]$ the \emph{peelings} of $P$.
Observe that $\inte H[N_H[\inte H[P]]] = \inte H[P]$, 
$\bd H[N_H[\inte H[P]]] \subseteq \bd H[P]$, and $\bd H[P] \setminus \bd H[N_H[\inte H[P]]]$ equals the peelings of $P$.

In the description, for brevity, we use $\emptyset$ to denote not only an empty set,
but also an empty graph. Moreover, for a $t$-boundaried graph $(G,\bdfun)$
and a set $Y \subseteq V(G)$, we somewhat abuse the notation and write
$(G \setminus Y,\bdfun)$
for the $t$-boundaried graph $(G \setminus Y,\bdfun|_{\bd G \setminus Y})$.

We assume we are given a \emph{nice} tree decomposition $(\Tree,\Tbag)$ of the
input graph $G$, and a labeling $\labfun:V(G) \to \{1,2,\ldots,t\}$ that
is injective on every bag.
We now describe how to conduct computations in each of the four types of nodes
in the tree decomposition $(\Tree,\Tbag)$.
In all cases, it will be clear from the description 
that all computations for a single node $w$ can be done
in time polynomial in $t$ and the number of states per one node,
and hence we do not further discuss the time complexity of the algorithm.

\medskip

\noindent\textbf{Leaf node.}
It is immediate that $\emptyset$ is feasible for every leaf node $w$
and every state $\state$ at $w$, and hence $T[w,\state] = 0$ is the correct value.

\medskip

\noindent\textbf{Introduce node.}
Consider now an introduce node $w$ with child $w'$, and the unique vertex
$v \in \Tbag(w) \setminus \Tbag(w')$.
Furthemore, consider a single state $\state = (\wX,\wG)$ at node $w$;
we are to compute $T[w,\state]$.

First, consider the case $v \in \wX$.
Then, it is straightforward to observe that $\state' := (\wX \setminus \{v\}, \wG)$
is a state for $w'$, and the families of feasible sets for state $\state$ and $w$
and for state $\state'$ and $w'$ are equal. Thus, $T[w,\state] = T[w',\state']$ and we are done.

Consider now the case $v \notin \wX$.
Define $\bdfun := \labfun|_{\Tbag(w) \setminus \wX}$ and $\bdfun' := \labfun|_{\Tbag(w') \setminus \wX}$.
We compute a family $\statefam$ of states at the node $w'$ and define $T[w,\state] = \min_{\state' \in \statefam} T[w',\state']$.
We are going to prove that 
\begin{enumerate}
\item for every $\state' \in \statefam$, and every $X$ that is feasible for $w'$ and $\state'$, $X$ is also feasible for $w$ and $\state$;\label{p:std:intro1}
\item for every $X$ that is feasible for $w$ and $\state$, there exists $\state' \in \statefam$ such that $X$ is also feasible for $w'$ and $\state'$.\label{p:std:intro2}
\end{enumerate}
Observe that such properties of $\statefam$ will imply the correctness of the computation of $T[w,\state]$. We now proceed to the construction of $\statefam$.

We iterate through all states $\state' = (\wX',\wG')$ for the node $w'$ and insert $\state'$ into $\statefam$ if the following conditions hold.
First, we require $\wX = \wX'$. Second, we construct a graph $\wG'_v$ by adding the vertex $v$ to $\wG'$,
and making $v$ adjacent to all vertices $u \in \bd \wG' = \Tbag(w') \setminus \wX$ for which $vu \in E(G)$.
Observe that $\wG'_v [\Tbag(w) \setminus \wX] = \wG[\Tbag(w) \setminus \wX]= G[\Tbag(w) \setminus \wX]$, and $v$ has exactly the same neighbourhood in $\wG'_v$ and in $\wG$.
To include $\state'$ in $\statefam$, we require that $(\wG'_v, \bdfun)$ is a witness-subgraph of $(\wG,\bdfun)$.

We first argue about Property~\ref{p:std:intro1}.
Let $\state' = (\wX, \wG') \in \statefam$ and let $X$ be feasible for $w'$ and $\state'$; we are to prove that $X$ is also feasible for $w$ and $\state$.
To this end, consider a set $Y \subseteq \Tbag(w) \setminus \wX$ and a $t$-slice $\slice$ such that $|Y| + |V(\slice)| \leq |V(H)|$.
Let $\pi$ be a $\slice$-subgraph in $(G[\Tall(w) \setminus (X \cup \wX \cup Y)], \bdfun)$.
Let $P = \pi^{-1}(\Tall(w'))$; note that in particular $v\notin \pi(P)$. Let $\slice'$ be the core slice of $P$, and let $Q$ be the peelings of $P$.
Observe that for every $a \in \bd P = Q \cup \bd \slice'$ we have $\pi(a) \neq v$ and, moreover, either $a \in \bd \slice$ or $\pi(a) \in N_{G[\Tall(w)]}(v)$.
In both cases, $\pi(a) \in \Tbag(w') \setminus (\wX \cup Y)$
and we may define $\bdfun_P(a) := \labfun(\pi(a))$. Note that with the labeling $\bdfun_P|_{\bd \slice'}$, the slice $\slice'$ becomes a $t$-slice,
   and $\pi|_{V(\slice')}$ is a $\slice'$-subgraph of $(G[\Tall(w') \setminus (X \cup \wX \cup Y \cup \pi(Q))],\bdfun')$.
Since $X$ is feasible for $w'$ and $\state'$, and $|Y| + |\pi(Q)| + |V(\slice')| \leq |Y| + |V(\slice)| \leq |V(H)|$,
we have that there exists a $\slice'$-subgraph $\pi'$ in $(\wG' \setminus (Y \cup \pi(Q)),\bdfun')$.
As no vertex of $Y$, $\pi(Q)$ nor $v$ belongs to the image of $\pi'$, and vertices of $Q$ are not adjacent to the vertices of $\inte \slice'$ in $H$,
a direct check shows that $\pi' \cup \pi|_{V(\slice) \setminus V(\slice')}$ is a $\slice$-subgraph of $(\wG'_v \setminus Y,\bdfun)$.
Recall that we require that $(\wG'_v,\bdfun)$ is a witness-subgraph of $(\wG,\bdfun)$.
Consequently, there exists a $\slice$-subgraph of $(\wG \setminus Y,\bdfun)$.
Since the choice of $Y$ and $\slice$ was arbitrary, $X$ is feasible for $w$ and $\state$.

For Property~\ref{p:std:intro2}, let $X$ be a feasible set for $w$ and $\state$.
Let $\wG'$ be the witness graph, whose existence is guaranteed by Lemma~\ref{lem:witness} for the graph $(G[\Tall(w') \setminus (X \cup \wX)], \bdfun')$.
Define $\state' = (\wX, \wG')$. Observe that $\state'$ is a valid state for $w'$.
By definition of $\wG'$, the set $X$ is feasible for $w'$ and $\state'$.
It remains to show that $\state' \in \statefam$, that is, that $(\wG'_v,\bdfun)$ is a witness-subgraph of $(\wG,\bdfun)$.

To this end, consider a set $Y \subseteq \Tbag(w) \setminus \wX$ and a $t$-slice $\slice$ such that $|Y| + |V(\slice)| \leq |V(H)|$.
Let $\pi$ be a $\slice$-subgraph in $(\wG'_v \setminus Y,\bdfun)$.
Similarly as before, let $P = \pi^{-1}(V(\wG'))$; again, note that $v\notin \pi(P)$. Let $\slice'$ be the core slice of $P$, and let $Q$ be the peelings of $P$.
Observe that for every $a \in \bd P = Q \cup \bd \slice'$ we have $\pi(a) \neq v$ and, moreover, either $a \in \bd \slice$ or $\pi(a) \in N_{\wG'_v}(v)$.
In both cases, $\pi(a) \in \Tbag(w') \setminus (\wX \cup Y)$,
and we may define $\bdfun_P(a) := \labfun(\pi(a))$. Note that with the labeling $\bdfun_P|_{\bd \slice'}$, the slice $\slice'$ becomes a $t$-slice,
   and $\pi|_{V(\slice')}$ is a $\slice'$-subgraph of $(\wG' \setminus (Y \cup \pi(Q)),\bdfun')$.
Since $|Y| + |\pi(Q)| + |V(\slice')| \leq |Y| + |V(\slice)| \leq |V(H)|$, 
by the properties of the witness graph guaranteed by Lemma~\ref{lem:witness}, we have that there exists a $\slice'$-subgraph $\pi'$ of
$(G[\Tall(w') \setminus (X \cup \wX \cup Y \cup \pi(Q))], \bdfun')$.
As no vertex of $Y$, $\pi(Q)$ nor $v$ belongs to the image of $\pi'$, and vertices of $Q$ are not adjacent to the vertices of $\inte \slice'$ in $H$,
a direct check shows that $\pi' \cup \pi|_{V(\slice) \setminus V(\slice')}$ is a $\slice$-subgraph of $(G[\Tall(w) \setminus (X \cup \wX \cup Y)],\bdfun)$.
Since $X$ is feasible for $w$ and $\state$, there exists a $\slice$-subgraph of $(\wG \setminus Y, \bdfun)$.
As the choice of $Y$ and $\slice$ was arbitrary, $(\wG'_v,\bdfun)$ is a witness-subgraph of $(\wG,\bdfun)$, and $\state' \in \statefam$.

This finishes the proof of the correctness of computations at an introduce node.

\medskip

\noindent\textbf{Forget node.}
Consider now a forget node $w$ with child $w'$, and the unique vertex $v \in \Tbag(w') \setminus \Tbag(w)$.
Let $\state = (\wX,\wG)$ be a state for $w$; we are to compute $T[w,\state]$.
Define $\bdfun := \labfun|_{\Tbag(w) \setminus \wX}$ and $\bdfun' := \labfun|_{\Tbag(w') \setminus \wX}$.

First, observe that $\state^v := (\wX \cup \{v\},\wG)$ is also a valid state for the node $w'$.
Note that $X \cup \{v\}$ is feasible for $\state$ if and only if $X$ is feasible for $w'$ and $\state^v$:
the question of feasibility of $X \cup \{v\}$ for $w$ and $\state$ and $X$ for $w'$ and $\state^v$ in fact inspects the same subgraph of $G$.
Consequently, we take $1+T[w',\state^v]$ as one candidate value for $T[w,\state]$.

In the remainder of the computations for the forget node we identify a family $\statefam$ of valid states for the node $w'$.
We prove that
\begin{enumerate}
\item for every $\state' \in \statefam$, and every $X$ that is feasible for $w'$ and $\state'$, $X$ is also feasible for $w$ and $\state$;\label{p:std:forget1}
\item for every $X$ that is feasible for $w$ and $\state$, and such that $v \notin X$, there exists $\state' \in \statefam$
such that $X$ is also feasible for $w'$ and $\state'$.\label{p:std:forget2}
\end{enumerate}
This claim will imply that 
$$T[w,\state] = \min(1+T[w',\state^v], \min_{\state' \in \statefam} T[w',\state']).$$

The family $\statefam$ is defined as follows. We iterate through all valid states $\state' = (\wX',\wG')$ for the node $w'$.
First, we require $\wX' = \wX$. Second, we define the graph $\wG'_v$ as the graph $\wG'$ with the label $\labfun(v)$ of the node $v$
forgotten, that is, $\wG'_v$ and $\wG'$ are equal as simple graphs and $\bd \wG'_v = \bd \wG' \setminus \{v\} = \Tbag(w) \setminus \wX$.
To include $\state'$ in $\statefam$, we require that $(\wG'_v,\bdfun)$ is a witness-subgraph of $(\wG,\bdfun)$.

For Property~\ref{p:std:forget1}, 
let $\state' = (\wX, \wG') \in \statefam$ and let $X$ be feasible for $w'$ and $\state'$; we are to prove that $X$ is also feasible for $w$ and $\state$.
To this end, consider a set $Y \subseteq \Tbag(w) \setminus \wX$ and a $t$-slice $\slice$ such that $|Y| + |V(\slice)| \leq |V(H)|$.
Let $\pi$ be a $\slice$-subgraph in $(G[\Tall(w) \setminus (X \cup \wX \cup Y)], \bdfun)$.
Note that $\pi$ is also a $\slice$-subgraph in $(G[\Tall(w') \setminus (X \cup \wX \cup Y)], \bdfun')$.
Since $X$ is feasible for $w'$ and $\state'$, there exists a $\slice$-subgraph $\pi'$ in $(\wG' \setminus Y,\bdfun')$.
As $\slice$ does not use the label $\labfun(v)$, by the definition of $\wG'_v$, $\pi'$
is also a $\slice$-subgraph of $(\wG'_v \setminus Y,\bdfun)$.
Since $(\wG'_v,\bdfun)$ is a witness-subgraph of $(\wG,\bdfun)$, there exists a $\slice$-subgraph
in $(\wG\setminus Y,\bdfun)$.
Since the choice of $Y$ and $\slice$ was arbitrary, $X$ is feasible for $w$ and $\state$ and Property~\ref{p:std:forget1} is proven.

For Property~\ref{p:std:forget2}, let $X$ be a feasible set for $w$ and $\state$, and assume $v \notin X$.
Let $\wG'$ be the witness graph, whose existence is guaranteed by Lemma~\ref{lem:witness}, for the graph $(G[\Tall(w') \setminus (X \cup \wX)],\bdfun')$.
Define $\state' = (\wX, \wG')$. Observe that $\state'$ is a valid state for $w'$.
By definition of $\wG'$, the set $X$ is feasible for $w'$ and $\state'$.
It remains to show that $\state' \in \statefam$, that is, that $(\wG'_v,\bdfun)$ is a witness-subgraph of $(\wG,\bdfun)$.

To this end, consider a set $Y \subseteq \Tbag(w) \setminus \wX$ and a $t$-slice $\slice$ such that $|Y| + |V(\slice)| \leq |V(H)|$.
Let $\pi$ be a $\slice$-subgraph in $(\wG'_v \setminus Y,\bdfun)$.
By the definition of $\wG'_v$, $\pi$ is also a $\slice$-subgraph of $(\wG' \setminus Y,\bdfun')$.
By the properties of the witness graph of Lemma~\ref{lem:witness}, there exists a $\slice$-subgraph $\pi'$ of $(G[\Tall(w') \setminus (X \cup \wX \cup Y)],\bdfun')$.
As $\slice$ does not use the label $\labfun(v)$, $\pi'$ is also a $\slice$-subgraph of $(G[\Tall(w) \setminus (X \cup \wX \cup Y)],\bdfun)$.
Since $X$ is feasible for $w$ and $\state$, there exists a $\slice$-subgraph of $(\wG,\bdfun)$.
As the choice of $Y$ and $\slice$ was arbitrary, $(\wG'_v,\bdfun)$ is a witness-subgraph of $(\wG,\bdfun)$ and, consequently, $\state' \in \statefam$.

This finishes the proof of the correctness of the computations at the forget node.

\medskip

\noindent\textbf{Join node.}
Let $w$ be a join node with children $w_1$ and $w_2$, and let $\state = (\wX,\wG)$
be a state for $w$. Define $\bdfun = \labfun|_{\Tbag(w) \setminus \wX}$.

Our goal is to define a family $\statefam$ of pairs of states $(\state_1,\state_2)$
such that $\state_i$ is a valid state for the node $w_i$ for $i=1,2$, and:
\begin{enumerate}
\item for every $(\state_1,\state_2) \in \statefam$, and every
pair of sets $X_1,X_2$, such that $X_i$ is feasible for $w_i$ and $\state_i$, $i=1,2$, the set $X := X_1 \cup X_2$ is feasible for $w$ and $\state$;\label{p:std:join1}.
\item for every $X$ that is feasible for $w$ and $\state$, there exists a pair $(\state_1,\state_2) \in \statefam$
such that the set $X_i := X \cap \Tdown(w_i)$ is feasible for $w_i$ and $\state_i$, $i=1,2$.\label{p:std:join2}
\end{enumerate}
This claim will imply that
$$T[w,\state] = \min_{(\state_1,\state_2) \in \statefam} T[w_1,\state_1] + T[w_2,\state_2].$$

The family $\statefam$ is defined as follows. We iterate through all pairs $(\state_1,\state_2)$
such that $\state_i = (\wX_i,\wG_i)$ is a valid state for the node $w_i$, $i=1,2$.
First, we require $\wX = \wX_1 = \wX_2$.
Second, we define the graph $\wG_1 \oplus \wG_2$ as a disjoint union of the graphs $\wG_1$
and $\wG_2$ with the boundaries $\bd \wG_1 = \bd \wG_2 = \Tbag(w) \setminus \wX$ identified.
To include $(\state_1,\state_2)$ into $\statefam$, we require that $(\wG_1 \oplus \wG_2, \bdfun)$ is a witness-subgraph of $(\wG,\bdfun)$.

For Property~\ref{p:std:join1}, let $(\state_1,\state_2) \in \statefam$, $\state_i = (\wX,\wG_i)$ for $i=1,2$.
Recall that $X_i$ is feasible for $w_i$ and $\state_i$, $i=1,2$, and $X = X_1 \cup X_2$; we are to prove that $X$ is feasible for $w$ and $\state$.
To this end, let $Y \subseteq \Tbag(w) \setminus \wG$ and let $\slice$ be a $t$-slice such that $|Y| + |V(\slice)| \leq |V(H)|$.
Assume there exists a $\slice$-subgraph $\pi$ in $(G[\Tall(w) \setminus (X \cup \wX \cup Y)],\bdfun)$.

First, let $P_1 = \pi^{-1}(\Tall(w_1))$.
Let $\slice_1$ be the core slice of $P_1$, and let $Q_1$ be the peelings of $P_1$.
Observe that for every $a \in \bd P_1 = Q_1 \cup \bd \slice_1$ we have $\pi(a) \in \Tbag(w) \setminus (\wX \cup Y)$, and
we may define $\bdfun_1(a) := \labfun(\pi(a))$. Note that with the labeling $\bdfun_1|_{\bd \slice_1}$, the slice $\slice_1$ becomes a $t$-slice,
   and $\pi|_{V(\slice_1)}$ is a $\slice_1$-subgraph of $(G[\Tall(w_1) \setminus (X_1 \cup \wX \cup Y \cup \pi(Q_1))],\bdfun)$.
Since $X_1$ is feasible for $w_1$ and $\state_1$, and $|Y| + |\pi(Q_1)| + |V(\slice_1)| \leq |Y| + |V(\slice)| \leq |V(H)|$,
we have that there exists a $\slice_1$-subgraph $\pi_1$ in $(\wG_1 \setminus (Y \cup \pi(Q_1)),\bdfun)$.
As no vertex of $Y \cup \pi(Q_1)$ belongs to the image of $\pi_1$, and vertices of $Q_1$ are not adjacent to the vertices of $\inte \slice_1$ in $H$,
a direct check shows that $\pi_\circ := \pi_1 \cup \pi|_{V(\slice) \setminus V(\slice_1)}$ is a $\slice$-subgraph of $((\wG_1 \oplus G[\Tall(w_2) \setminus (X_2 \cup \wX)]) \setminus Y,\bdfun)$.

We now perform the same operation for $w_2$ and $\pi_\circ$. That is, let $P_2 = \pi_\circ^{-1}(\Tall(w_2))$.
Let $\slice_2$ be the core slice of $P_2$, and let $Q_2$ be the peelings of $P_2$.
Again, we have $\pi_\circ(\bd P_2) \subseteq \Tbag(w) \setminus (\wX \cup Y)$
and we define $\bdfun_2 := \labfun|_{\bd P_2}$, turning $\slice_2$ into a $t$-slice.
The mapping $\pi_\circ|_{V(\slice_2)}$ is a $\slice_2$-subgraph of $(G[\Tall(w_2) \setminus (X_2 \cup \wX \cup Y \cup \pi(Q_2))],\bdfun)$.
Since $X_2$ is feasible for $w_2$ and $\state_2$, and $|Y| + |\pi(Q_2)| + |V(\slice_2)| \leq |Y| + |V(\slice)| \leq |V(H)|$,
we have that there exists a $\slice_2$-subgraph $\pi_2$ in $(\wG_2 \setminus (Y \cup \pi(Q_2)),\bdfun)$.
Similarly as before, a direct check shows that $\pi_2 \cup \pi_\circ|_{V(\slice) \setminus V(\slice_2)}$ is a $\slice$-subgraph of $((\wG_1 \oplus \wG_2) \setminus Y,\bdfun)$.

Recall that we require that $(\wG_1 \oplus \wG_2,\bdfun)$ is a witness-subgraph of $(\wG,\bdfun)$.
Consequently, there exists a $\slice$-subgraph in $(\wG \setminus Y,\bdfun)$.
Since the choice of $Y$ and $\slice$ was arbitrary, $X$ is feasible for $w$ and $\state$.

For Property~\ref{p:std:join2}, let $X$ be a feasible set for $w$ and $\state$.
For $i=1,2$, recall that $X_i = X \cap \Tdown(w_i)$ and
let $\wG_i$ be the witness graph whose existence is guaranteed by Lemma~\ref{lem:witness}
for the graph $(G[\Tall(w_i) \setminus (X_i \cup \wX)],\bdfun)$.
Observe that $\state_i := (\wX,\wG_i)$ is a valid state for the node $w_i$.
Moreover, by definition, $X_i$ is feasible for $w_i$ and $\state_i$.
To finish the proof of Property~\ref{p:std:join2}, it suffices to show
that $(\state_1,\state_2) \in \statefam$, that is, $(\wG_1 \oplus \wG_2,\bdfun)$ is a witness-subgraph of $(\wG,\bdfun)$.

To this end, consider a set $Y \subseteq \Tbag(w) \setminus \wX$ and a $t$-slice $\slice$ such that $|Y| + |V(\slice)| \leq |V(H)|$.
Let $\pi$ be a $\slice$-subgraph in $(\wG_1 \oplus \wG_2,\bdfun)$.
Let $P_1 = \pi^{-1}(V(\wG_1))$, let $\slice_1$ be the core slice of $P_1$, and let $Q_1$ be the peelings of $P_1$.
Observe that for every $a \in \bd P_1 = Q_1 \cup \bd \slice_1$ we have $\pi(a) \in \Tbag(w) \setminus (\wX \cup Y)$,
and we may define $\bdfun_1(a) := \labfun(\pi(a))$. Note that with the labeling $\bdfun_1|_{\bd \slice_1}$, the slice $\slice_1$ becomes a $t$-slice,
and $\pi|_{V(\slice_1)}$ is a $\slice_1$-subgraph of $(\wG_1 \setminus (Y \cup \pi(Q_1)),\bdfun)$.
As $|Y| + |\pi(Q_1)| + |V(\slice_1)| \leq |Y| + |V(\slice)| \leq |V(H)|$, 
by the properties of the witness graph of Lemma~\ref{lem:witness},
there exists a $\slice_1$-subgraph $\pi_1'$ of $(G[\Tall(w_1) \setminus (X_1 \cup \wX \cup Y \cup \pi(Q_1))],\bdfun)$.
As no vertex of $Y \cup \pi(Q_1)$ belongs to the image of $\pi_1'$, and vertices of $Q_1$ are not adjacent to the vertices of $\inte \slice_1$ in $H$,
a direct check shows that $\pi_\circ := \pi_1' \cup \pi|_{V(\slice) \setminus V(\slice_1)}$
is a $\slice$-subgraph of $((G[\Tall(w_1) \setminus (X_1 \cup \wX)] \oplus \wG_2) \setminus Y,\bdfun)$.

We now perform a similar operation for $w_2$ and $\pi_\circ$.
Let $P_2 = \pi_\circ^{-1}(V(\wG_2))$, let $\slice_2$ be the core slice of $P_2$, and let $Q_2$ be the peelings of $P_2$.
Since $\pi_\circ(\bd P_2) \subseteq \Tbag(w) \setminus (\wX \cup Y)$, we define $\bdfun_2 := \labfun|_{\bd P_2}$, turning $\slice_2$ into a $t$-slice.
The mapping $\pi|_{V(\slice_2)}$ is a $\slice_2$-subgraph of $(\wG_2 \setminus (Y \cup \pi(Q_2)),\bdfun)$.
As $|Y| + |\pi(Q_2)| + |V(\slice_2)| \leq |Y| + |V(\slice)| \leq |V(H)|$, 
by the properties of the witness graph of Lemma~\ref{lem:witness},
there exists a $\slice_2$-subgraph $\pi_2'$ of $(G[\Tall(w_2) \setminus (X_2 \cup \wX \cup Y \cup \pi(Q_2))],\bdfun)$.
Similarly as before, 
a direct check shows that $\pi_2' \cup \pi_\circ|_{V(\slice) \setminus V(\slice_2)}$
is a $\slice$-subgraph of $(G[\Tall(w) \setminus (X \cup \wX \cup Y)],\bdfun)$.

Since $X$ is feasible for $w$ and $\state$, there exists a $\slice$-subgraph of $(\wG,\bdfun$).
As the choice of $Y$ and $\slice$ was arbitrary, $(\wG_1 \oplus \wG_2,\bdfun)$ is a witness-subgraph of $(\wG,\bdfun)$ and $(\state_1,\state_2) \in \statefam$.

This finishes the proof of the correctness of computations at a join node.

\medskip

This finishes the description of the dynamic programming algorithm
for Theorem~\ref{thm:std:algo}, and concludes its proof.

\section{Discussion on special cases of \shit{}}\label{sec:std-discussion}

As announced in the introduction, we now discuss a few special cases
of \shit{}.

\subsection{Hitting a path}
First, let us consider $H$ being a path, $H = P_h$ for some $h \geq 3$. 
Note that  $\msep(P_h) = 1$, while $\mnei(P_h) = 2$ for $h\geq 5$.
Observe that in the dynamic programming algorithm of the previous section
we have that $G[\Tall(w) \setminus (X \cup X_w)]$ does not contain an $H$-subgraph
and, hence, the witness graph obtained through Lemma~\ref{lem:witness}
does not contain an $H$-subgraph as well.
However, graphs exluding $P_h$ as a subgraph have treedepth (and hence treewidth as well) bounded by $h$
(since any their depth-first search tree has depth bounded by $h$). 
Using this insight, we can derive the following
improvement of Lemma~\ref{lem:witness}, that
improves the running time of Theorem~\ref{thm:std:algo}
to $2^{\Oh(t \log t)} |V(G)|$ for $H$ being a path.
\begin{lemma}\label{lem:witness-path}
Assume $H$ is a path.
Then, for any $t$-boundaried graph $(G,\bdfun)$ that does not contain an $H$-subgraph,
  there exists a witness graph as in Lemma~\ref{lem:witness} with $\Oh(t)$ vertices and edges.
\end{lemma}
\begin{proof}
In this proof, by \emph{witness graph} we mean any graph $(\wG,\bdfun)$
that satisfies the requirements of Lemma~\ref{lem:witness},
for fixed input $(G,\bdfun)$ and a graph $H$ being a path.

A witness graph $(\wG,\bdfun)$ is \emph{minimal} if, for every $v \in V(\wG) \setminus \bd \wG$
the graph $(\wG \setminus \{v\},\bdfun)$ is not a witness graph.
We claim that, in the case $H = P_h$, every minimal witness graph $\wG$ has only $\Oh(t)$
vertices and edges, assuming $G$ does not contain any $H$-subgraph.
Clearly, this claim will prove Lemma~\ref{lem:witness-path}.

Fix a minimal witness graph $(\wG,\bdfun)$.
Since $H=P_h$ is connected, observe the following: any connected component $C$ of $\wG$ needs to contain
at least one vertex of $\bd \wG$, as otherwise $(\wG \setminus C,\bdfun)$
is a witness graph as well, a contradiction.
For each connected component $C$ of $\wG$, we fix some depth-first search
spanning tree $T_C$ of $\wG[C]$, rooted in some vertex $r_C \in \bd \wG \cap C$.
Since $G$ (and thus $\wG$) does not contain an $H$-subgraph,
the depth of $T_C$ is less than $h = \Oh(1)$. Since $T_C$ is a depth-first search tree, all the edges of $C$ connect vertices that are in ancestor-descendant relation in $T_C$; in other words, $C$ is a subgraph of the ancestor-descendant closure of $T_C$.

Consider any $v \in C$.
Let $v_1,v_2,\ldots,v_s$ be the children of $v$ in the tree $T_C$
and let $T_i$ be the subtree of $T_C$ rooted at $v_i$.
Without loss of generality, assume that there exists $0 \leq r \leq s$ such that
a tree $T_i$ contains a vertex of $\bd \wG$ if and only if $i > r$.
We claim that $r \leq h^4 = \Oh(1)$.

Let us first verify that this claim proves that $|C| = \Oh(|\bd \wG \cap C|)$, and hence summing up through the components $C$ will conclude the proof of Lemma~\ref{lem:witness-path}.
Since $T_C$ has depth less that $h$, and there are $|\bd \wG \cap C|$ vertices in $C$ that belong to $\bd \wG$, we infer that in $T_C$ there are at most $h|\bd \wG \cap C|$ vertices $u$ such that the subtree rooted at $u$ contains some vertex of $\bd \wG$. However, if we consider any vertex $u$ such that subtree rooted at $u$ does not contain any vertex of $\bd \wG$, then this subtree must be of size $\Oh(h^{4h})$: its depth is less than $h$, and every vertex has at most $h^4$ children. Therefore, the tree $T_C$ contains at most $h|\bd \wG \cap C|$ vertices $u$ whose subtrees contain vertices of $\bd \wG$, and each of these vertices has at most $h^4$ subtrees rooted at its children that are free from $\bd \wG$, and thus have size $\Oh(h^{4h})$. We infer that $|V(C)|=|V(T_C)|\leq h|\bd \wG \cap C|\cdot (1+h^4\cdot \Oh(h^{4h}))=\Oh(|\bd \wG \cap C|)$, because $h$ is considered a constant.

In order to prove that $r\leq h^4$, let us mark some of the trees $T_i$, $1 \leq i \leq r$.
Let $r_C = w_0, w_1, \ldots, w_p = v$ be the path between $r_C$ and $v$ in the tree $T_C$; note that we have $p<h$.
For any $0 \leq a < b \leq p$, and for any $3 \leq l < h$,
mark any $h$ trees $T_i$, $1 \leq i \leq r$, such that $G[\{w_a,w_b\} \cup T_i]$
contains a path between $w_a$ and $w_b$ with exactly $l$ vertices.
(If there is less than $h$ such trees, we mark all of them.)
For any $0 \leq a \leq p$, and for any $2 \leq l < h$, 
mark any $h$ trees $T_i$, $1 \leq i \leq r$, such that $G[\{w_a\} \cup T_i]$
contains a path with endpoint $w_a$ and with exactly $l$ vertices.
(Again, if there is less than $h$ such trees, we mark all of them.)
We claim that all trees $T_i$, $1 \leq i \leq r$ are marked;
as we have marked less than $h^4$ trees, this would conclude the proof of the lemma.

By contradiction, without loss of generality assume that $T_1$ is not marked.
We claim that $(\wG \setminus T_1,\bdfun)$ is a witness graph as well.
Consider any $Y \subseteq V(G)$ and $t$-slice $\slice$ such that $|Y| + |V(\slice)| \leq |V(H)|$ and assume there
exists a $\slice$-subgraph in $(G \setminus Y,\bdfun)$.
By the definition of a witness graph, there exists a $\slice$-subgraph
in $(\wG \setminus Y,\bdfun)$.
Let $\pi$ be such a subgraph that contains a minimum possible number of vertices of $T_1$
in its image. We claim that there are none, and $\pi$ is a
$\slice$-subgraph in $(\wG\setminus (Y \cup T_1),\bdfun)$ as well.
Assume the contrary. Since $T_1$ does not contain any vertex of $\bd \wG$,
there exists a subpath $H'$ of $\slice$, with $|V(H')| = l$ for some $l<h$,
such that either (a) both endpoints of $H'$ are mapped by $\pi$
to some $w_a,w_b$, $0 \leq a < b \leq p$, and the internal vertices
of $H'$ are mapped to $T_1$; or (b) one endpoint of $H'$
is mapped by $\pi$ to some $w_a$, $0 \leq a \leq p$, and all other vertices
of $H'$ are mapped to $T_1$. Since $T_1$ was not marked, we infer that in both cases there exist at least $h$ marked trees $T_i$ that also contain such a subpath $H'$.
Since the union of $Y$ and the image of $\pi$ has cardinality at most $h$, we infer that there exists a tree $T_i$ that was marked for the same
choice of $a,b$ and $l$ in case (a), or $a$ and $l$ in case (b),
and, furthermore, no vertex of $T_i$ is contained neither in $Y$ nor in the image of $\pi$.
We modify $\pi$ by remapping all vertices of $V(H') \cap \pi^{-1}(T_1)$
to the corresponding vertices of $T_i$.
In this manner we obtain a $\slice$-subgraph of $(\wG \setminus Y,\bdfun)$
with strictly less vertices in $T_1$ in its image, a contradiction
to the choice of $\pi$.
Hence, $\pi$ does not use any vertex of $T_1$, and is a
$\slice$-subgraph of $(\wG \setminus (Y \cup T_1), \bdfun)$.
This concludes the proof.
\maybeqed\end{proof}

\begin{corollary}
For every positive integer $h$, the \shitarg{P_h} problem can be solved in time 
$2^{\Oh(t \log t)} |V(G)|$ on a graph $G$ of treewidth $t$.
\end{corollary}

\subsection{Hitting pumpkins}

Now let us consider $H = K_{2,h}$ for some $h \geq 2$ (such a graph is sometimes called a \emph{pumpkin} in the literature).
Observe that $\mnei(K_{2,h}) = \msep(K_{2,h}) = h$.
On the other hand, we note the following.
\begin{lemma}\label{lem:witness-K2h}
Assume $H = K_{2,h}$ for some $h \geq 2$.
If the witness graph given by Lemma~\ref{lem:witness} does not admit an $H$-subgraph, then it
has $\Oh(t^2)$ vertices and edges.
\end{lemma}
\begin{proof}
Let $V(H) = \{a_1,a_2,b_1,b_2,\ldots,b_h\}$ where $A := \{a_1,a_2\}$
and $B:= \{b_1,b_2,\ldots,b_h\}$ are bipartition classes of $H$.
Note that there are only two types of proper chunks in $H$:
$N_H[a_i]$, $i=1,2$ and $N_H[b_j]$, $1 \leq j \leq h$.
Hence, one can easily verify that in the construction of the witness graph $\wG$ of Lemma~\ref{lem:witness}
every vertex $v \in \wG \setminus \bd \wG$ has at least two neighbours in $\bd \wG$,
and $\wG \setminus \bd \wG$ is edgeless.
Then we have $|N_\wG(v) \cap \bd \wG|\leq 2\binom{|N_\wG(v) \cap \bd \wG|}{2}$ for each $v \in V(\wG)\setminus \bd \wG$.

However, since the constructed witness graph $\wG$ does not admit an $H$-subgraph,
each two vertices $v_1,v_2 \in \bd \wG$ have less
than $h$ common neighbours in $\wG$, as otherwise there is a $H$-subgraph
in $\wG$ on vertices $v_1$, $v_2$ and $h$ vertices of $N_\wG(v_1) \cap N_\wG(v_2)$.
Hence
\begin{equation}\label{eq:K2h}
\sum_{v \in V(\wG)\setminus \bd \wG} \binom{|N_\wG(v) \cap \bd \wG|}{2} \leq (h-1)\binom{|\bd \wG|}{2} \leq (h-1)\binom{t}{2}.
\end{equation}
Consequently, there are at most $2(h-1)\binom{t}{2}$ edges of $\wG$
with exactly one endpoint in $\bd \wG$, whereas there are at most $\binom{t}{2}$ edges in $\wG[\bd \wG]$. The lemma follows.
\maybeqed\end{proof}

Lemma~\ref{lem:witness-K2h} together with the dynamic programming of Section~\ref{ss:std-algo} imply that \shitarg{K_{2,h}} can be solved in time $2^{\Oh(t^2\log t)}|V(G)|$, in spite of the fact that $\mnei(K_{2,h}) = \msep(K_{2,h}) = h$. 

\begin{corollary}
For every positive integer $h$, the \shitarg{K_{2,h}} problem can be solved in time 
$2^{\Oh(t^2 \log t)} |V(G)|$ on a graph $G$ of treewidth $t$.
\end{corollary}

We now show that a slight modification of $K_{2,h}$ enables us to prove a much higher lower bound. For this, let us consider a graph $H_h$ for $h \geq 2$
defined as $K_{2,h}$ with triangles attached to both degree-$h$ vertices (see Figure~\ref{fig:Hh}).
Note that $\msep(H_h) = \mnei(H_h) = h$.
One may view $H_h$ as $K_{2,h}$ with some symmetries broken, so that
the proof of Lemma~\ref{lem:witness-K2h} does not extend to $H_h$.
We observe that the lower bound proof of Theorem~\ref{thm:intro:lb:col}
works, with small modifications, also for the case of \shitarg{H_h}.
As the proof for this special case is slightly simpler than the one of Theorem~\ref{thm:intro:lb:col},
we present it first here to give intuition to the reader.
\begin{theorem}[Theorem~\ref{thm:lb:Hh} restated]\label{thm:lb:Hh2}
Unless ETH fails, for every $h \geq 2$ there does
not exist an algorithm that, given a \shitarg{H_h} instance $G$
and a tree decomposition of $G$ of width $t$, resolves $G$
in time $2^{o(t^h)} |V(G)|^{\Oh(1)}$.
\end{theorem}

\begin{proof}
We denote the vertices of $H_h$ as in Figure~\ref{fig:Hh}.

\begin{figure}[tb]
\centering
\includegraphics{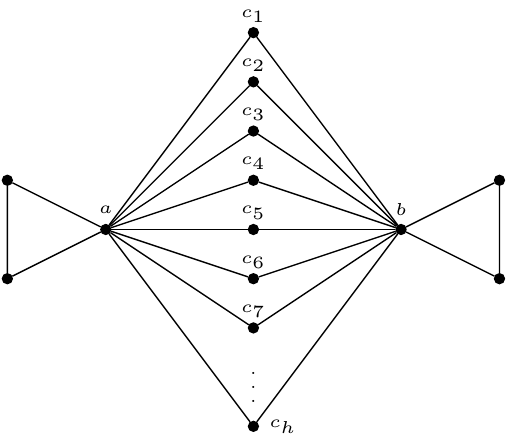}
\caption{The graph $H_h$, equal to $K_{2,h}$ with two triangles added to the degree-$h$ vertices.}
\label{fig:Hh}
\end{figure}

We first define the following basic operation for the construction.
By \emph{attaching a copy of $H_h$ at vertices $u$ and $v$} we mean the following:
we introduce a new copy of $H_h$ into the constructed graph, and identify $u$ with the copy 
of the vertex $a$, and $v$ with the copy of the vertex $b$.
The intutive idea behind attaching a copy of $H_h$ at $u$ and $v$ is that it forces
the solution to take $u$ or $v$.

Assume we are given as input a clean $3$-CNF formula $\Phi$
with $n$ variables and $m$ clauses.
We are to construct a \shitarg{H_h} instance $G$ with a budget bound $k$
and a tree decomposition of $G$ of width $\Oh(n^{1/h})$,
such that $G$ admits a solution of size at most $k$ if and only if $\Phi$
is satisfiable. This construction, together with Lemma~\ref{lem:preprocess} and Theorem~\ref{thm:spars}, proves the statement of the theorem.

Let $s$ be the smallest positive integer such that $s^h \geq 3n$. Observe that
$s = \Oh(n^{1/h})$. We start our construction by introducing
a set $M$ of $sh$ vertices $w_{j,i}$, $1 \leq j \leq s$, $1 \leq i \leq h$.
The set $M$ is the central part of the constructed graph $G$.
In particular, 
in our reduction each connected component of $G \setminus M$ will be of constant size,
yielding immediately the promised tree decomposition of $G$ of width $\Oh(n^{1/h})$.

To each clause $C$ of $\Phi$, and to each literal $l$ in $C$, assign a function
$f_{C,l}: \{1,2,\ldots,h\}
  \to \{1,2,\ldots,s\}$ such that $f_{C,l} \neq f_{C',l'}$ for $(C,l) \neq (C',l')$.
Observe that this is possible due to the assumption $s^h \geq 3n$ and the fact
that $\Phi$ is clean.

For each variable $x$ of $\Phi$, proceed as follows.
For each clause $C$ that contains $x$ in a literal $l \in \{x,\neg x\}$,
we introduce a new vertex $a_{x,C,l}$ and make it adjacent to all
vertices $w_{f_{C,l}(i),i}$ for $1 \leq i \leq h$.
Let $a_{x,C_1,l}$, $a_{x,C_2,\neg l}$ and $a_{x,C_3,l}$ be the three
vertices introduced;
recall that $x$ appears exactly three times in $\Phi$, twice positively and once
negatively or twice negatively and once positively.
Moreover, we introduce a fourth dummy vertex $a_x$.
Finally, we attach a copy of $H_h$ to the following four pairs of vertices:
$(a_{x,C_1,l}, a_{x,C_2,\neg l})$,
$(a_{x,C_3,l}, a_{x,C_2,\neg l})$,
$(a_{x,C_1,l}, a_x)$, and
$(a_{x,C_3,l}, a_x)$.
Let $D_x$ be the set of vertices constructed for variable $x$.
Observe that, for every variable $x$,
we have constucted four $H_h$-subgraphs, and there are exactly two ways
to hit them with only two vertices: either we take
$\{a_{x,C_1,l}, a_{x,C_3,l}\}$ into the solution or 
$\{a_{x,C_2,\neg l}, a_x\}$ into the solution.
Moreover, any solution to \shitarg{H_h} on the constructed graph needs
to take at least two vertices of $D_x$.

For each clause $C$ of $\Phi$, proceed as follows.
For each literal $l$ in $C$, introduce a new vertex $b_{C,l}$
and make it adjacent to 
all vertices $w_{f_{C,l}(i),i}$ for $1 \leq i \leq h$.
For each two different literals $l_1,l_2$ in $C$, attach
a copy of $H_h$ at vertices $b_{C,l_1}$ and $b_{C,l_2}$.
Let $D_C$ be the set of vertices constructed for the clause $C$.
Observe that, for every clause $C$ that contains $r_C$ literals
(recall $2 \leq r_C \leq 3$), we have constructed a number of 
$H_h$-subgraphs, we need at least $r_C-1$ vertices of the solution
to hit them, and, without loss of generality, we may assume that
any solution that contains only $r_C-1$ vertices of $D_C$
actually contains all but one of the vertices $b_{C,l}$.

We set the budget
$$k = 2n + \sum_{C \in \Phi} r_C-1 = 5n-m.$$
Observe that this budget is tight: any solution $X$ to \shitarg{H_h}
on $G$ of size at most $k$ needs to contain exactly two vertices in each $D_x$,
exactly $r_C-1$ vertices in each $D_C$ and, consequently,
is of size exactly $k$ and does not contain any more vertices.

It remains to argue about the correctness of this construction.
The crucial observation is that there are only few $H_h$-subgraphs in $G$:
the vertices $a_{x,C,l}$, $a_x$ and $b_{C,l}$ are the only vertices of $G$
that have degree at least $h$ and, at the same time, are contained in some triangle
in $G$. With this observation, a direct check shows that, apart from
the copies of $H_h$ introduced explicitely in the construction,
the only other copies are ones with vertices $a$ and $b$ mapped
to $a_{x,C,l}$ and $b_{C,l}$ for every clause $C$, every literal $l$ in $C$,
and $x$ being the variable of $l$.

In one direction, let $\phi$ be a satisfying assignment of $\Phi$.
Construct a solution $X$ as follows.
For each variable $x$, include into $X$ all vertices $a_{x,C,l}$ for which
$l$ is satisfied in $\phi$. Moreover, include also the vertex $a_x$
if only one vertex $a_{x,C,l}$ has been included in the previous step.
For each clause $C$, pick a literal $l$ that satisfies it,
and include into $X$ all vertices $b_{C,l'}$ for $l' \neq l$.
This concludes the description of the set $X$. 
Clearly, $|X| = k$. Moreover, observe that, for each clause $C$
and literal $l$ in $C$, $l \in \{x,\neg x\}$, we have that either
$a_{x,C,l}$ or $b_{C,l}$ belongs to $X$. With the previous insight
into the family of $H_h$-subgraphs of $G$, this implies that $X$
hits all $H_h$-subgraphs of $G$.

In the second direction, let $X$ be such a set of at most $k$ vertices
that hits all $H_h$-subgraphs in $G$.
By the discussion on tightness of the budget, $X$ contains exactly $2$ 
vertices in each set $D_x$ and exactly $r_C-1$ vertices in each set $D_C$, and no more vertices.
Define an assignment $\phi$ as follows:
for each variable $x$, we set $\phi(x)$ so that $a_{x,C,l} \in X$ if and only
if $l$ is evaluated to true by $\phi$. Note that this is a valid definition
by the construction of $G[D_x]$. 
To show that $\phi$ satisfies $\Phi$, consider a clause $C$.
By the budget bounds, there exists a vertex $b_{C,l} \notin X$.
Recall that there exists an $H_h$-subgraph of $G$ with $a$ mapped to $a_{x,C,l}$
and $b$ mapped to $b_{C,l}$, where $x$ is the variable of $l$.
As $X$ hits all $H_h$-subgraphs of $G$, $a_{x,C,l} \in X$.
Hence, $\phi$ sets $l$ to true and hence satisfies $C$.
This concludes the proof.
\maybeqed\end{proof}

%Observe that the proof of Theorem~\ref{thm:lb:Hh} does not need to assume
%that $h$ is a constant. Thus, we obtain the following interesting double-exponential
%lower bound.
%\begin{corollary}[Corollary~\ref{cor:lb:Hh:double-exp} restated]\label{cor:lb:Hh:double-exp2}
%Unless ETH fails, there does not exist an algorithm that, given a graph $G$ with a tree decomposition of width $t$,
%and an integer $h = \Oh(\log |V(G)|)$, finds in time $2^{2^{o(t)}} |V(G)|^{\Oh(1)}$
%the minimum size of a set that hits all $H_h$-subgraphs of $G$.
%\end{corollary}
%\begin{proof}
%We proceed with the same construction as in Theorem~\ref{thm:lb:Hh}, but set
%$h = \lceil \log n \rceil$ at the begining.
%Then $s = \Oh(1)$ and, consequently, $|M| = \Oh(\log n)$ and each connected component of $G \setminus M$
%has size $\Oh(\log n)$.
%Hence, it is immediate to provide a tree decomposition of $G$ of width $\Oh(\log n)$.
%An algorithm that finds in $2^{2^{o(t)}} |V(G)|^{\Oh(1)}$ the minimum size of a set that hits all $H_h$-subgraphs of $G$
%would in fact resolve the input formula $\Phi$ in time $2^{o(n)}$, a contradiction to ETH.
%\maybeqed\end{proof}

\section{Algorithms for \cshit{}}\label{sec:algo-col}

In this section we develop algorithmic upper bounds
for \cshit{} on graphs of bounded treewidth.
We start with a simple observation that essentially reduces
the problem to $H$ being a connected graph.

\begin{lemma}\label{lem:cshit:conn}
Let $(G,\col)$ be a \cshit{} instance $(G,\col)$.
Then, a set $X \subseteq V(G)$ hits all 
$\col$-$H$-subgraphs of $G$
if and only if
there exists a connected component $C$ of $H$
such that, if we define $V_C = \col^{-1}(C) \subseteq V(G)$,
then $X \cap V_C$ hits all $\col|_{V_C}$-$C$-subgraphs of $G[V_C]$.
\end{lemma}
\begin{proof}
The right-to-left implication is immediate. In the other direction, 
by contradiction, assume that for each connected component $C$
of $H$ there exists a $\col|_{V_C}$-$C$-subgraph $\pi_C$ in $G[V_C]$ that is not hit by $X$.
Then, as each vertex of $H$ has its own color,
$\bigcup_C \pi_C$ is a $\col$-$H$-subgraph in $G$ that is not hit by $X$.
\maybeqed\end{proof}

Hence, Lemma~\ref{lem:cshit:conn} allows us to solve \cshit{} problem
only for connect graphs $H$: in the general case, we solve \cshitarg{C}
on $(G[V_C],\col|_{V_C})$ for each connected component $C$ of $H$.
In the remainder of this section we consider only connected graphs $H$.

We now resolve two simple special cases: when $H$ is a path or a clique.

\begin{theorem}\label{thm:cshit:poly}
\cshit{} is polynomial-time solvable in the case
when $H$ is a path.
\end{theorem}

\begin{proof}
Let $h = |V(H)|$ and let $a_1,a_2,\ldots,a_h$ be the consecutive vertices on the path $H$.

For an input $(G,\col)$ to \cshit{}, construct an auxiliary directed graph $G'$ as follows.
Take $V(G') = V(G) \cup \{s,t\}$, where $s$ and $t$ are two new vertices.
For each $1 \leq i < h$ and every edge $uv \in E(G)$
with $\col(u) = a_i$ and $\col(v) = a_{i+1}$, add an arc $(u,v)$ to $G'$.
Moreover, for each $u \in \col^{-1}(a_1)$ add an arc $(s,u)$ and for each 
$u \in \col^{-1}(a_h)$ add an arc $(u,t)$.
Observe the family of $\col$-$H$-subgraphs of $G$ is in one-to-one correspondance
with directed paths from $s$ to $t$ in the graph $G'$.
Hence, to compute a set $X$ of minimum size that hits all $\col$-$H$-subgraphs of $G$
it suffices to compute a minimum cut between $s$ and $t$ in the graph $G'$. This can be done in polynomial time using any maximum flow algorithm.
\maybeqed\end{proof}

\begin{theorem}\label{thm:cshit:clique}
A \cshit{} instance $(G,\col)$
can be solved in time $2^{\Oh(t)} |V(G)|$ in the case
when $H$ is a clique and $t$ is the treewidth of $G$.
\end{theorem}

\begin{proof}
We perform an absolutely standard dynamic programming algorithm, essentially using the folklore fact
that any clique in $G$ needs to be completely contained in a bag of the decomposition.

Recall that we may assume that we are given a nice tree decomposition
$(\Tree,\Tbag)$ of $G$, of width \emph{less} than $t$.
For a node $w$, a \emph{state} is a set $\wX \subseteq \Tbag(w)$.
For a node $w$ and a state $\wX$, we say 
that a set $X \subseteq \Tdown(w)$ is \emph{feasible}
if $G[\Tall(w) \setminus (\wX \cup X)]$ does not contain any $\col$-$H$-subgraph.
Define $T[w,\wX]$ to be the minimum size of a feasible set for $w$ and $\wX$.
Note that $T[\mathtt{root}(\Tree),\emptyset]$ is the answer to the input
\cshit{} instance. We now show how to compute the values $T[w,\wX]$
in a bottom-up manner in the tree decomposition $(\Tree,\Tbag)$.

\medskip

\noindent\textbf{Leaf node.} Observe that
$\emptyset$ is the unique valid state for a leaf node $w$,
and $T[w,\emptyset] = 0$.

\medskip

\noindent\textbf{Introduce node.}
Consider now an introduce node $w$ with child $w'$, and a unique vertex
$v \in \Tbag(w) \setminus \Tbag(w')$.
Furthemore, consider a single state $\wX$ at node $w$.
If $v \in \wX$, then clearly a set $X$ is feasible for $w$ and $\wX$
if and only if it is feasible for $w'$ and $\wX \setminus \{v\}$.
Hence, $T[w,\wX] = T[w',\wX \setminus \{v\}]$ in this case.

Consider now the case $v \notin \wX$. If there is a
$\col$-$H$-subgraph in $G[\Tbag(w) \setminus \wX]$ then clearly no
set is feasible for $w$ and $\wX$, and $T[w,\wX] = +\infty$.
Otherwise, since $H$ is a clique and $v$ does not have any neighbor in $\Tdown(w)$, there is no $\col$-$H$-subgraph of
$G[\Tall(w)]$ that uses both $v$ and a vertex of $\Tdown(w)$.
Consequently, $T[w,\wX] = T[w',\wX]$ in the remaining case.

\medskip

\noindent\textbf{Forget node.}
Consider now a forget node $w$ with child $w'$,
and a unique vertex $v \in \Tbag(w') \setminus \Tbag(w)$.
Let $\wX \subseteq \Tbag(w)$ be any state.
We claim that $T[w,\wX] = \min(T[w',\wX], 1+T[w',\wX \cup \{v\}])$.

In one direction, it suffices to observe that, for any set $X$ feasible for $w$ and $\wX$,
if $v \in X$, then $X \setminus \{v\}$ is feasible
for $T[w',\wX \cup \{v\}]$, and otherwise, if $v \notin X$,
then $X$ is feasible for $T[w',\wX]$.
In the other direction, note that any feasible set for $w'$ and $\wX$
is also feasible for $w$ and $\wX$, whereas for every feasible set $X$ for $w'$ and $\wX \cup \{v\}$,
$X \cup \{v\}$ is feasible for $w$ and $\wX$.

\medskip

\noindent\textbf{Join node.}
Let $w$ be a join node with children $w_1$ and $w_2$, and let $\wX$ be a state for $w$.
We claim that $T[w,\wX] = T[w_1,\wX] + T[w_2,\wX]$. 

Indeed, note that, since $H$ is a clique, any $\col$-$H$-subgraph of $G[\Tall(w) \setminus \wX]$
has its image entirely contained in $\Tall(w_1)$ or entirely contained in $\Tall(w_2)$.
Consequently, if $X_i$ is feasible for $w_i$ and $\wX$, for $i=1,2$,
then $X_1 \cup X_2$ is feasible for $w$ and $\wX$. In the other direction,
it is straightforward that for every feasible set $X$ for $w$ and $\wX$,
$X \cap \Tdown(w_i)$ is feasible for $w_i$ and $\wX$, $i=1,2$.
\maybeqed\end{proof}

We now move to the proof of Theorem~\ref{thm:intro:cshit:algo}.

\begin{theorem}[Theorem~\ref{thm:intro:cshit:algo} restated]\label{thm:cshit:algo}
A \cshit{} instance $(G,\col)$ 
can be solved in time $2^{\Oh(t^{\msep(H)})} |V(G)|$
in the case when $H$ is connected and is not a clique, 
   where $t$ is the treewidth of $G$.
\end{theorem}
\begin{proof}
Recall that we may assume that we are given a nice tree decomposition
$(\Tree,\Tbag)$ of $G$, of width \emph{less} than $t$,
 and a labeling $\labfun: V(G) \to \{1,2,\ldots,t\}$ that is injective on each bag.

We now define a state that will be used in the dynamic programming algorithm.
%A \emph{potential chunk} is a separator $t$-chunk $\chunk[A]$ with interior $A$.
A \emph{state} at node $w \in V(\Tree)$ is a pair $(\wX,\chunkfam)$
where $\wX \subseteq \Tbag(w)$ and $\chunkfam$ is a family of separator $t$-chunks,
where each chunk $\chunk$ in $\chunkfam$:
\begin{enumerate}
\item uses only labels of $\labfun^{-1}(\Tbag(w) \setminus \wX)$;
\item the mapping $\pi: \bd \chunk \to \Tbag(w) \setminus \wX$ that
maps a vertex of $\bd \chunk$ to a vertex with the same label is a homomorphism
of $H[\bd \chunk]$ into $G$ (in particular, it respects colors).
\end{enumerate}
Observe that, as $|\bd \chunk| \leq \msep(H)$ for any separator chunk $\chunk$,
there are $\Oh(t^{\msep(H)})$ possible separator $t$-chunks, and hence
$2^{\Oh(t^{\msep(H)})}$ possible states for a fixed node $w$.

The intuitive idea behind a state is that, for node $w \in V(\Tree)$ and state $(\wX,\chunkfam)$
we investigate the possibility of the following: for a solution $X$ we are looking for, it holds that
   $\wX = X \cap \Tbag(w)$ and
the family $\chunkfam$ is exactly the set of possible separator
chunks of $H$ that are subgraphs of $G \setminus X$, where the subgraph relation is defined as on
$t$-boundaried graphs and $G \setminus X$ is equiped with $\bd G \setminus X = \Tbag(w) \setminus X$
and labeling $\labfun|_{\Tbag(w) \setminus X}$.
The difficult part of the proof is to show that this information is sufficient, in particular,
it suffices to keep track only of separator chunks, and not all chunks of $H$.
We emphasize here that the intended meaning of the set $\chunkfam$ is that it represents
separator chunks present in the entire $G \setminus X$, not $G[\Tall(w)] \setminus X$.
That is, to be able to limit ourselves only to separator chunks, we need to encode some prediction
for the future in the state. This makes our dynamic programming algorithm
rather non-standard.

Let us proceed to the formal definition of the dynamic programming table.
For a bag $w$ and a state $\state = (\wX,\chunkfam)$ at $w$ we define the graph $G(w,\state)$
as follows. We first take the $t$-boundaried graph $(G[\Tall(w)] \setminus \wX, \labfun|_{\Tbag(w) \setminus \wX})$,
and then, for each chunk $\chunk \in \chunkfam$ we add a disjoint copy of $\chunk$ to $G(w,\state)$ and identify the pairs
of vertices with the same label in $\bd \chunk$ and in $\Tbag(w) \setminus \wX$.
Note that $G[\Tall(w)] \setminus \wX$ is an induced subgraph of $G(w,\state)$: by the properties
of elements of $\chunkfam$, no new edge has been introduced between two vertices of $\Tbag(w)$.
We make $G(w,\state)$ a $t$-boundaried graph in a natural way: $\bd G(w,\state) = \Tbag(w) \setminus \wX$
with labeling $\labfun|_{\bd G(w,\state)}$.

For each bag $w$ and for each state $\state = (\wX,\chunkfam)$
we say that a set $X \subseteq \Tdown(w)$ is \emph{feasible} if
$G(w,\state) \setminus X$ does not contain any $\col$-$H$-subgraph, and for
any separator $t$-chunk $\chunk$ of $H$, if there is a $\chunk$-subgraph in 
$G(w,\state) \setminus X$ then $\chunk \in \chunkfam$.
We would like to compute the value $T[w,\state]$ that equals a minimum size of a feasible set $X$.

We remark that a reverse implication to the one in the definition of a feasible set $X$ 
(i.e., all chunks of $\chunkfam$ are present in $G(w,\state) \setminus X$) 
is straightforward for any $X$:
we have explicitly glued all chunks of $\chunkfam$ into $G(w,\state)$.
Hence, $T[w,\state]$ asks for a minimum
size of set $X$ whose deletion not only deletes all $\col$-$H$-subgraphs,
but also makes $\chunkfam$ a ``fixed point'' of an operation
of gluing $G[\Tall(w)] \setminus (X \cup \wX)$ along the boundary $\Tbag(w) \setminus \wX$.

Observe that, if $\Tbag(w) = \emptyset$, there is only one state $(\emptyset,\emptyset)$ valid
for the node $w$. Hence, $T[\texttt{root}(\Tree),(\emptyset,\emptyset)]$ is the 
answer to \cshit{} on $(G,\col)$.
In the rest of the proof we focus on computing values $T[w,\state]$ in a bottom-up
fashion in the tree $\Tree$.

\medskip

\noindent\textbf{Leaf node.} As already observed, there is only one valid state for empty bags.
Hence, for each leaf node $w$, we may set $T[w,(\emptyset,\emptyset)] = 0$.

\medskip

\noindent\textbf{Introduce node.} Consider now an introduce node $w$ with child $w'$, and the unique vertex
$v \in \Tbag(w) \setminus \Tbag(w')$.
Furthermore, consider a single state $\state = (\wX,\chunkfam)$ at node $w$.
There are two cases, depending on whether $v \in \wX$ or not.
Observe that if $v \in \wX$ then $\state' = (\wX \setminus \{v\}, \chunkfam)$ is a valid state
for the node $w'$ and, moreover, $G(w,\state) = G(w',\state')$. 
Hence, $T[w,\state] = T[w',\state']$.

The second case, when $v \notin \wX$, is more involved.
Let $\chunkfam' \subseteq \chunkfam$ be the set of all these chunks in $\chunkfam$
that do not use the label $\labfun(v)$. Observe that $\state' = (\wX,\chunkfam')$ is a valid
state for the node $w'$. In what follows we prove that $T[w,\state] = T[w',\state']$
unless some corner case happens.

Consider first the graph $G^\circ := G(w,\state) \setminus \Tdown(w)$ (i.e., we glue all chunks 
of $\chunkfam$, but only to $G[\Tbag(w) \setminus \wX]$ as opposed to of
$G[\Tall(w) \setminus \wX]$; this corresponds to taking the deletion set $X$ maximal possible) with $\bd G^\circ = \Tbag(w)\setminus \wX$.
If there exists an $\col$-$H$-subgraph of $G^\circ$
or a  separator $t$-chunk $\chunk$ that is a subgraph
of $G^\circ$, but does not belong to $\chunkfam$, then clearly there is no
feasible set $X$ for $w$ and $\state$.
In this case we set $T[w,\state] = +\infty$.
Observe that the graph $G^\circ$ has size $t^{\Oh(\msep(H))}$, and hence we can check
if the aforementioned corner case happens in time polynomial in $t$.

We now show that in the remaining case $T[w,\state] = T[w',\state']$.
Consider first a feasible set $X$ for node $w$ and state $\state$.
We claim that $X$ is also feasible for $w'$ and $\state'$; note that $\Tdown(w) = \Tdown(w')$.
Observe that $G(w',\state')$ is a subgraph of $G(w,\state)$, thus, in particular, $G(w',\state') \setminus X$ does not contain any $\col$-$H$-subgraph.
Let $\chunk$ be any separator $t$-chunk that is a subgraph of $G(w',\state') \setminus X$.
By the previous argument, we have that $\chunk$ is a subgraph of $G(w,\state) \setminus X$ as well.
Since $X$ is feasible for $w$ and $\state$, we have that $\chunk \in \chunkfam$. Since $\chunk$ is a subgraph of $G(w',\state') \setminus X$, it does not use the label $\labfun(v)$
and we infer $\chunk \in \chunkfam'$. Thus, $X$ is feasible for $w'$ and $\state'$, and, consequently, $T[w,\state] \geq T[w',\state']$.

In the other direction, consider a feasible set $X$ for node $w'$ and state $\state'$.
We would like to show that $X$ is feasible for $w$ and $\state$.
We start with the following structural observation about minimal separators.
\begin{myclaim}\label{cl:H-sep}
Let $\chunk$ be a separator chunk in $H$.
Let $a_1,a_2 \in \inte \chunk$ and assume $Z \subseteq V(\chunk)$ is an $a_1a_2$-separator in $\chunk$.
Then there exists a separator chunk $\chunk'$ that (a) contains $a_1$ or $a_2$ in its interior, (b) whose vertex set is a proper subset of the vertex set of $\chunk$,
and (c) such that $\bd \chunk' \subseteq Z \cup \bd \chunk$
\end{myclaim}
\begin{proof}
Since $\chunk$ is a separator chunk, there exists a connected component $B$ of $H \setminus \bd \chunk$ that
is vertex-disjoint with $\chunk$, and such that $N_H(B) = \bd \chunk$.
Since $Z$ is an $a_1a_2$-separator in $\chunk$, and $a_1,a_2 \in \inte \chunk$, we have that $Z' := Z \cup \bd \chunk$ is an $a_1a_2$-separator in $H$.
Then we may find a set $S \subseteq Z'$ that is a minimal $a_1a_2$-separator in $H$. For $i=1,2$, let $A_i$ be the connected component of $H \setminus S$ that contains $a_i$.
Since $A_1$ and $A_2$ are vertex-disjoint, and $H[B]$ is connected, $B$ can be contained only in one of these sets.
Without loss of generality we may assume $A_1 \cap B = \emptyset$.
As $N_H(B) = \bd \chunk$, we infer that $A_1 \cap \bd \chunk = \emptyset$ and, consequently, $A_1 \subseteq \inte \chunk$.
Moreover, $a_2 \notin N_H[A_1]$, and hence $A_1 \subsetneq \inte \chunk$.
Since $N_H(A_1) = S$, $\chunk' := \chunk[A_1]$ is a separator chunk that satisfies all the requirements of the claim.
\cqed\end{proof}

We first prove the condition for feasibility with respect to separator $t$-chunks.
\begin{myclaim}\label{cl:introduce-chunk}
For any separator $t$-chunk $\chunk$, if there exists a $\col$-$\chunk$-subgraph in $G(w,\state) \setminus X$, then $\chunk \in \chunkfam$.
\end{myclaim}
\begin{proof}
By contradiction, assume now that there exists a separator $t$-chunk $\chunk$ such that $\chunk \notin \chunkfam$,
but there exists a $\col$-$\chunk$-subgraph $\pi$ in $G(w,\state) \setminus X$.
Without loss of generality assume that $\chunk$ has minimum possible number of vertices, and, for fixed $\chunk$,
the image of $\pi$ contains minimum possible number of vertices of $\Tdown(w) \cup (V(G(w,\state)) \setminus V(G(w',\state')))$.
Since $\chunk$ is a separator chunk, there exists a connected component $B$ of $H \setminus \bd \chunk$ that is vertex-disjoint
with $\chunk$ and such that $N_H(B) = \bd \chunk$. Let $b \in B$ be any vertex.

If the image of $\pi$ does not contain any vertex of $\Tdown(w)$, then $\pi$ is a $\col$-$\chunk$-subgraph in $G^\circ$, and $\chunk \in \chunkfam$ by the previous steps.
Hence, there exists a vertex $a_1 \in \inte \chunk$ such that $\pi(a_1) \in \Tdown(w)$.
Let $A_1$ be the connected component of $H[\pi^{-1}(\Tdown(w))]$ that contains $a_1$.
Observe that $N_H(A_1)$ separates $a_1$ from $b$ in $H$ and $N_H(A_1) \subseteq V(\chunk)$.
Hence, in $H$ there exists a minimal $a_1b$-separator $S \subseteq N_H(A_1)$. Let $A_1'$ be the connected
component of $H \setminus S$ that contains $A_1$. Since $N_H(B) = \bd \chunk$ and $S \subseteq N_H(A_1)\subseteq V(\chunk)$, we infer that $A_1' \subseteq \inte \chunk$.
By minimality of $S$ we have $N_H(A_1') = S$, and hence $\chunk' = \chunk[A_1']$ is a separator
chunk. Moreover, $\pi(N_H(A_1'))=\pi(S) \subseteq \pi(N_H(A_1)) \subseteq \Tbag(w')$, and we may equip $\chunk'$
with a labeling $x \mapsto \labfun(\pi(x))$ for any $x \in S$, constructing a separator $t$-chunk.

Assume first $\chunk' \neq \chunk$. Then $\chunk'$ has strictly less vertices than $\chunk$ (as $A_1' \subseteq \inte \chunk$).
By the choice of $\chunk$, we have $\chunk' \in \chunkfam$. 
Since $v$ does not have any neighbors in $\Tdown(w)$, we have that $v\notin \pi(N_H(A_1))$, so in particular $v\notin \pi(S)$. Hence in fact $\pi(S) \subseteq \Tbag(w')$, and the label $\labfun(v)$ is not used in $\chunk'$. Therefore $\chunk' \in \chunkfam'$.
Consequently, we can modify $\pi$
by remapping the vertices of $A_1'$ to the copy of $\chunk'$ that has been glued into $G(w',\state')$, obtaining
a $\chunk$-subgraph of $G(w,\state) \setminus X$ with strictly less vertices in $\Tdown(w) \cup (V(G(w,\state)) \setminus V(G(w',\state')))$, a contradiction. Note here that the color constraints ensure that the vertices we remap $A_1'$ to are not used by $\pi$, and thus the remapped $\pi$ is still injective.

We are left with the case $\chunk' = \chunk$, that is, $A_1' = \inte \chunk'$. In particular, this implies $v \notin \pi(\bd \chunk)$.
If the image of $\pi$ does not contain any vertex of $V(G(w,\state)) \setminus V(G(w',\state'))$, then
$\pi$ is a $\chunk$-subgraph in $G(w',\state') \setminus X$
and, by the feasibility of $X$ for $w'$ and $\state'$, we have $\chunk \in \chunkfam' \subseteq \chunkfam$, a contradiction.
Hence, there exists a vertex $a_2 \in V(\chunk)$ such that $\pi(a_2) \notin V(G(w',\state'))$.

Observe that $Z := \pi^{-1}(\Tbag(w'))\cap \chunk$ is a separator between $a_1$ and $a_2$ in $\chunk$. We apply Claim~\ref{cl:H-sep} and obtain a chunk $\chunk''$.
Note that $\pi(\bd \chunk'') \subseteq \Tbag(w')$, as $v \notin \pi(\bd \chunk)$. 
Hence, we may treat $\chunk''$ as a separator $t$-chunk with labeling $x \mapsto \labfun(\pi(x))$,
and a restriction of $\pi$ is a $\chunk''$-subgraph in $G(w,\state) \setminus X$.
Since $\chunk''$ has strictly less vertices than $\chunk$, by the choice of $\chunk$ we have $\chunk'' \in \chunkfam$.
Moreover, $v \notin \pi(Z \cup \bd \chunk)$ and, consequently, $\labfun(v)$ is not used in $\chunk''$ and $\chunk'' \in \chunkfam'$.
Hence, we may modify $\pi$ by remapping all vertices of $\inte \chunk''$ to the copy of $\chunk''$, glued onto $\Tbag(w')$ in the process of constructing $G(w',\state')$; again, the color constraints ensure that this remapping preserves injectivity of $\pi$.
In this manner we obtain
a $\col$-$\chunk$-subgraph of $G(w,\state) \setminus X$ with strictly less vertices in $\Tdown(w) \cup (V(G(w,\state)) \setminus V(G(w',\state')))$, a contradiction.
\cqed\end{proof}

We now move to the second property of a feasible set.
\begin{myclaim}\label{cl:introduce-H}
There are no $\col$-$H$-subgraphs in $G(w,\state)\setminus X$.
\end{myclaim}
\begin{proof}
By contradiction, assume there exists a $\col$-$H$-subgraph $\pi$ in $G(w,\state) \setminus X$.
Without loss of generality pick $\pi$ such that minimizes the number of vertices of $\pi(V(H))$ that belong to $\Tdown(w)$.
As $G^\circ$ does not contain any $\col$-$H$-subgraph, there exists $a \in V(H)$ such that $\pi(a) \in \Tdown(w)$.
Since $X$ is feasible for $w'$ and $\state'$, there exists $b \in V(H)$ such that $\pi(b) \notin G(w',\state')$.
Observe that $\Tbag(w') \setminus \wX$ separates $\Tdown(w) \setminus X$ from $G(w,\state) \setminus \Tall(w')$ in the graph $G(w,\state) \setminus X$.
Hence, there exists a minimal $ab$-separator $S$ in $H$ such that $\pi(S) \subseteq \Tbag(w')$. Since $H$ is connected, $S \neq \emptyset$.
Let $A$ be the connected component of $H \setminus S$ that contains $a$.
Note that $\chunk[A]$ is a separator chunk in $H$ with $\bd \chunk[A] = N_H(A) = S$. 
Define $\bdfun: S \to \{1,2,\ldots,t\}$ as $\bdfun(x) = \labfun(\pi(x))$ for any $x \in S$.
With this labeling, $\chunk[A]$ becomes a separator $t$-chunk and $\pi|_{N_H[A]}$ is a $\col$-$\chunk$-subgraph in $G(w,\state) \setminus X$.
By Claim~\ref{cl:introduce-chunk}, $\chunk \in \chunkfam$.
Hence, we can modify $\pi$ by remapping all vertices of $\inte \chunk$ to the copy of $\chunk$ that has been
glued into $G(w,\state)$ in the process of its construction; again, the color constraints ensure that this remapping preserves injectivity of $\pi$.
In this manner we obtain a $\col$-$H$-subgraph in $G(w,\state) \setminus X$ with strictly less verties of $\pi(V(H)) \cap \Tdown(w)$, a contradiction to the choice of $\pi$.
\cqed\end{proof}

Claims~\ref{cl:introduce-chunk} and~\ref{cl:introduce-H} conclude the proof of the correctness of computations at an introduce node.

\medskip

\noindent\textbf{Forget node.}
Consider now a forget node $w$ with child $w'$, and a unique vertex $v \in \Tbag(w') \setminus \Tbag(w)$.
Let $\state = (\wX,\chunkfam)$ be a state for $w$; we are to compute $T[w,\state]$.
We shall identify a (relatively small) family $\statefam$ of states for $w'$ such that
\begin{enumerate}
\item for any $\state' = (\wX',\chunkfam') \in \statefam$ and for any set $X'$ feasible for $w'$ and $\state'$, the set $X' \cup (\{v\} \cap \wX')$ is feasible for $w$ and $\state$;
\item for any set $X$ feasible for $w$ and $\state$, there exists $\state' = (\wX',\chunkfam') \in \statefam$ such that $X \cap \{v\} = \wX' \cap \{v\}$ and 
$X \setminus \{v\}$ is feasible for $w'$ and $\state'$.
\end{enumerate}
Using such a claim, we may conclude that
$$T[w,\state] = \min_{\state' = (\wX',\chunkfam') \in \statefam} T[w',\state'] + |\wX' \cap \{v\}|.$$
We now proceed to the construction of $\statefam$.

First, observe that $\state' := (\wX \cup \{v\}, \chunkfam)$ is a valid state for $w'$ and, moreover, $G(w',\state') = G(w,\state) \setminus \{v\}$.
Hence, for any $X' \subseteq \Tdown(w')$ we have that $X'$ is feasible for $\state'$ if and only if $X' \cup \{v\}$ is feasible for $\state$.
Thus, it is safe to include $\state'$ in the family $\statefam$ (it satisfies the first property of $\statefam$) and it fulfils the second property
for all feasible sets $X$ containing $v$.

Second, we add to $\statefam$ all states $\state' = (\wX',\chunkfam')$ for the node $w'$ that satisfy the following: $\wX' = \wX$
and $\chunkfam$ is exactly the family of these chunks $\chunk\in \chunkfam'$ that do not use the label $\labfun(v)$.
Observe that for any such state, $G(w',\state')$ is a supergraph of $G(w,\state)$, with the only difference
that in $G(w',\state')$ the vertex $v$ receives label $\labfun(v)$. We infer that any set $X'$ that is feasible for $w'$ and $\state'$
is also feasible for $w$ and $\state$. 

This finishes the description of the family $\statefam$. It remains to argue that for every set $X$ that is feasible for $w$ and $\state$ and does not contain $v$,
there exists a state $\state'$ added in the second step such that $X$ is feasible for $w'$ and $\state'$.

To this end, for a fixed such $X$ define $\chunkfam'$ as follows: a separator $t$-chunk $\chunk$ belongs to $\chunkfam'$ if and only
if there is a $\chunk$-subgraph in a $t$-boundaried graph $G^1 := (G(w,\state) \setminus X, \labfun|_{\Tbag(w')})$.
We emphasize that this definition 
differs from the standard definition of a $t$-boundaried graph $G(w,\state) \setminus X$ on the vertex $v$: $v \in \bd G^1$ and it has label $\labfun(v)$.
Observe that, since $G^1$ differs from $G(w,\state) \setminus X$ (treated as a $t$-boundaried graph) only on the labeling of $v$,
we have that $\chunkfam$ is exactly the family of chunks of $\chunkfam'$ that do not use the label $\labfun(v)$.
Moreover, note that $\state' := (\wX,\chunkfam')$ is a valid state for $w'$, and hence $\state' \in \statefam$.
We now argue that $X$ is feasible for $w'$ and $\state'$.

To this end, consider any (possibly $t$-boundaried) subgraph $H'$ of $H$ and a $\col$-$H'$-subgraph $\pi$ in $G(w',\state') \setminus X$.
Observe that $G^1$ is a subgraph of $G(w',\state') \setminus X$, as $\chunkfam \subseteq \chunkfam'$.
Pick $\pi$ that minimizes the number of vertices of $\pi(V(H'))$ that do not belong to $V(G^1)$.
We claim that there is no such vertex.
Assume otherwise, and let $a \in V(H')$ be such that $\pi(a) \notin V(G^1)$. Thus, $\pi(a)$ lies in the interior
of some chunk $\chunk \in \chunkfam'$ that was glued onto $\Tbag(w')$ in the process of constructing $G(w',\state')$.
By the definition of $\chunkfam'$, there exists a $\col$-$\chunk$-subgraph $\pi_\chunk$ in $G^1$.
Define $\pi'$ as follows $\pi'(c) = \pi_\chunk(c)$ if $c$ belongs to $\chunk$, and $\pi'(c) = \pi(c)$ otherwise. Again, the color constraints ensure that $\pi'$ is injective.
Observe that $\pi'$ is also a $\col$-$H'$-subgraph in $G(w',\state')$ with strictly smaller number of vertices in $\pi(V(H')) \setminus V(G^1)$, a contradition to the choice of $\pi$. 

We infer that, for any (possibly $t$-boundaried) subgraph $H'$ of $H$, there exists a $\col$-$H'$-subgraph $\pi$ in $G(w',\state') \setminus X$ if and only if it exists in $G^1$.
This implies that $X$ is feasible for $w'$ and $\state'$, and concludes the description of the computation in the forget node.

\medskip

\noindent\textbf{Join node.}
Let $w$ be a join node with children $w_1$ and $w_2$, and let $\state = (\wX,\chunkfam)$ be a state for $w$.
Observe that $\state$ is a valid state for $w_1$ and $w_2$ as well.
We claim that $T[w,\state] = T[w_1,\state] + T[w_2,\state]$. 

To prove it, first consider a set $X$ feasible for $w$ and $\state$. Define $X_i = X \cap \Tdown(w_i)$ for $i=1,2$; note that $X = X_1 \uplus X_2$.
Observe that $G(w_i,\state) \setminus X_i$ is a subgraph of $G(w,\state)$ for every $i=1,2$.
Hence, $X_i$ is feasible for $w_i$ and $\state$ and, consequently, $T[w,\state] \geq T[w_1,\state] + T[w_2,\state]$.

In the other direction, let $X_1$ be a feasible set for $w_1$ and $\state$, and let $X_2$ be a feasible set for $w_2$ and $\state$.
We claim that $X := X_1 \cup X_2$ is feasible for $w$ and $\state$; note that such a claim would conclude the description of the computation at a join node.
To this end, let $\pi$ be a $\col$-$H'$-subgraph in $G(w,\state) \setminus X$, where $H'=H$ or $H'$ is a some separator $t$-chunk in $H$.
In what follows we prove, by induction on $|V(H')|$,
that there exists a $\col$-$H'$-subgraph in $G(w_1,\state) \setminus X_1$ or in $G(w_2,\state) \setminus X_2$. Note that this claim is sufficient to prove that $X$ is feasible
for $w$ and $\state$.

Fix $H'$. Without loss of generality, pick $\pi$ that minimizes the number of vertices in $\pi(V(H')) \cap \Tdown(w)$. 
If there exists $i \in \{1,2\}$ such that $\pi(V(H'))$ does not contain any vertex of $\Tdown(w_i)$, then $\pi$ is a $\col$-$H'$-subgraph of $G(w_{3-i},\state) \setminus X_{3-i}$, and we are done.
Hence, assume for each $i \in \{1,2\}$ there exists $a_i \in V(H')$ such that $\pi(a_i) \in \Tdown(w_i)$.
Observe that $Z := \pi^{-1}(\Tbag(w))$ separates $a_1$ from $a_2$ in $H'$.

First, we focus on the case $H' = \chunk$ being a separator $t$-chunk.
For a fixed chunk $\chunk$, we apply Claim~\ref{cl:H-sep} to the vertices $a_1,a_2$ and the set $Z$, obtaining a chunk $\chunk'$.
By the inductive assumption, there exists a $\col$-$\chunk'$-subgraph in $G(w_i,\state) \setminus X_i$ for some $i=1,2$.
Hence, $\chunk' \in \chunkfam$, and we may modify $\pi$ by redirecting the vertices of $\inte \chunk'$ to the copy of $\chunk'$
that has been glued into $G(w_i,\state)$ in the process of its construction.
In this manner, we obtain a $\col$-$H'$-subgraph of $G(w,\state) \setminus X$ with strictly less vertices of $\Tdown(w)$,
a contradiction to the choice of $\pi$.

Second, we focus on the case $H' = H$. 
Let $S \subseteq Z$ be a minimal $a_1a_2$-separator in $H$.
Let $A_1$ be the connected component of $H \setminus S$ that contains $a_1$.
Note that $\chunk[A_1]$ with the labeling $x \mapsto \labfun(\pi(x))$ for $x \in S$ is a separator $t$-chunk, 
and $\pi|_{N_H[A_1]}$ is a $\col$-$\chunk$-subgraph in $G(w,\state) \setminus X$.
By the induction hypothesis, there exists a $\col$-$\chunk$-subgraph in $G(w_i,\state) \setminus X_i$ for some $i=1,2$ and, consequently, $\chunk \in \chunkfam$.
Modify $\pi$ by redirecting the vertices of $\inte \chunk$ to the copy of $\chunk$ that has been glued into $G(w,\state)$ in the process of its construction.
In this manner we obtain a $\col$-$H$-subgraph in $G(w,\state) \setminus X$ with strictly less vertices in $\Tdown(w)$, a contradiction to the choice of $\pi$.

Thus, we have shown that $X$ is a feasible set for $w$ and $\state$, concluding the proof of the correctness of the computation at join node.
Since each node has $2^{\Oh(t^{\msep(H)})}$ states, and there are $\Oh(t|V(G)|)$ nodes in $\Tree$, the time bound of the computation follows.
This concludes the proof of Theorem~\ref{thm:cshit:algo}.
\maybeqed\end{proof}

\section{Lower bound for \cshit{}}\label{sec:lb-col}

In this section we prove a tight lower bound
for \cshit{}. The proofs are inspired by the approach of~\cite{cut-and-count-logic}
for the lower bond for \shitarg{C_\ell}.

In our constructions, we often use the following basic operation.
Let $(G,\col)$ be an $H$-colored graph constructed so far.
We pick some induced subgraph $H[Z]$ of $H$ and ``add a copy of $H[Z]$ to $G$''.
By this, we mean that we take a disjoint union of $G$ and $H[Z]$, and color
$H[Z]$ (extend $\col$ to the copy of $H[Z]$) naturally: a vertex $d \in H[Z]$ receives
color $d$.
After this operation, we often identify some vertices of the new copy $H[Z]$ with some
old vertices of $G$. However, we always identify pairs of vertices of the same color,
thus $\col$ is defined naturally after the identification.

We first start with a simple single-exponential lower bound
that describes the case when $H$ is a forest.

\begin{theorem}\label{thm:lb:col-vc}
Let $H$ be a graph that contains a connected component that is not a path.
Then, unless ETH fails, there does
not exist an algorithm that, given a \cshit{} instance $(G,\col)$
and a tree decomposition of $G$ of width $t$, resolves $(G,\col)$
in time $2^{o(t)} |V(G)|^{\Oh(1)}$.
\end{theorem}
\begin{proof}
We reduce from the well-known \textsc{Vertex Cover} problem.
We show a polynomial-time algorithm that, given a graph $G_0$, 
outputs a \cshit{} instance $(G,\col)$ together with a tree
decomposition of $G$ of width $|V(G_0)| + \Oh(1)$
such that the minimum possible size
of a vertex cover in $G_0$ equals the minimum possible size of a solution to $(G,\col)$ minus
$|E(G_0)|$.
As a $2^{o(n)}$-time algorithm for \textsc{Vertex Cover}
would contradict ETH~\cite{vc-subexp}, this will imply the statement of the theorem.

Let $C$ be a connected component of $H$ that is not a path.
It is easy to verify that in such a component there always exist at least three vertices
that are not cutvertices; let $a$, $b$ and $c$ be any three of them.

We construct an instance $(G,\col)$ as follows.
We start with $V(G) = V(G_0)$, with each vertex of $V(G_0)$ colored $a$.
Then, for each edge $uv \in E(G_0)$ we add three copies $C_{uv},C_{uv,u}$ and $C_{uv,v}$
of the graph $H[C]$
and identify the following pairs of vertices:
\begin{itemize}
\item the vertex $a$ in the copy $C_{uv,u}$ with the vertex $u$,
\item the vertex $a$ in the copy $C_{uv,v}$ with the vertex $v$,
\item the vertices $b$ in the copies $C_{uv,u}$ and $C_{uv}$, and
\item the vertices $c$ in the copies $C_{uv,v}$ and $C_{uv}$.
\end{itemize}
Define $D_{uv}$ to be the vertex set of all copies of $H[C]$ introduced for the edge $uv$.
Let $\pi_{uv}, \pi_{uv,u}$ and $\pi_{uv,v}$ be the (naturally induced)
injective homomorphisms from $H[C]$ to $G[C_{uv}], G[C_{uv,u}]$ and $G[C_{uv,v}]$, respectively.

Let $(G',\col')$ be the instance constructed so far, and observe
that all values of $\col'$ lie in $C$.

Finally, we add to $G$ a large number (at least $|E(G_0)| + |V(G_0)| + 1$) copies of $H \setminus C$.
This completes the description of the instance $(G,\col)$.
Observe that each connected component of $G \setminus V(G_0)$ is of size at most
$3|V(H)|-4 = \Oh(1)$.
Hence, it is straightforward to provide a tree decomposition of $G$ of width $|V(G_0)| + \Oh(1)$.

We now argue about the correctness of the reduction. First, let $Z$ be a vertex cover of $G_0$.
Define $X \subseteq V(G)$ as follows: first set $X := Z$ and then, for each edge $uv \in E(G_0)$,
if $u \in Z$ then add to $X$ the vertex $c$ in the copy $C_{uv}$
and otherwise add the vertex $b$ in the copy $C_{uv}$.
Clearly, $|X| = |Z| + |E(G_0)|$. We claim that $X$ hits all $\col'$-$H[C]$-subgraphs of $G'$,
and, consequently by Lemma~\ref{lem:cshit:conn}, hits all $\col$-$H$-subgraphs of $G$.

Let $\pi$ be any $\col'$-$H[C]$-subgraph of $G'$.
Recall that none of the vertices $a$, $b$ and $c$ is a cutvertex of $H[C]$.
As each vertex of $V(G_0)$ is colored $a$, at most one such vertex can be used in the image of $\pi$.
We infer that there exists
a single edge $uv \in E(G_0)$ such that $\pi(C) \subseteq D_{uv}$, as otherwise $\pi(a)$ is a cutvertex of $\pi(H[C])$.
Moreover, as the vertices $b$ and $c$ in the copy $C_{uv}$ are cutvertices of $G[D_{uv}]$,
we have that $\pi(C)$ is contained in one of the sets $C_{uv}$, $C_{uv,u}$ or $C_{uv,v}$.
However, each of this set is of size $|C|$, and each has non-empty intersection
with $X$. The claim follows.

In the other direction, let $X \subseteq V(G)$ be such that $X$ hits
all $\col$-$H$-subgraphs of $G$. We claim that there exists a vertex cover of $G_0$
of size at most $|X| - |E(G_0)|$. As $G$ contains a large number of copies of $H \setminus C$,
the set $X \cap V(G')$ needs to hit all $\col'$-$H[C]$-subgraphs of $G'$.

In particular, $X$ hits $\pi_{uv}$ for every $uv \in E(G_0)$. For every $uv \in E(G_0)$
pick one $x_{uv} \in X \cap C_{uv}$; note that $x_{uv}$ are pairwise distinct
as the sets $C_{uv}$ are pairwise disjoint.
Denote $Y = X \setminus \{x_{uv} : uv \in E(G_0)\}$. 
Construct a set $Z \subseteq V(G_0)$ as follows: for each $y \in Y$
\begin{enumerate}
\item if $y \in V(G_0)$, insert $y$ into $Z$;
\item if $y \in C_{uv,u} \setminus \{u\}$ for some $uv \in E(G_0)$, insert $u$ into $Z$;
\item if $y \in C_{uv,v} \setminus \{v\}$ for some $uv \in E(G_0)$, insert $v$ into $Z$; 
\item otherwise, do nothing.
\end{enumerate}
Observe that each $y \in Y$ gives rise to at most one vertex in $Z$. Hence,
$|Z| \leq |Y| = |X| - |E(G_0)|$. 
We finish the proof of the theorem by showing that
$Z$ is a vertex cover of $G_0$.

Consider any $uv \in E(G_0)$. The vertex $x_{uv}$ cannot be simultanously equal
to both the vertex $b$ and the vertex $c$ in the copy $C_{uv}$; by symmetry, assume
$x_{uv}$ does not equal the vertex $b$ in the copy $C_{uv}$.
Hence, $x_{uv}$ does not hit $\pi_{uv,u}$ and, consequently, there exists $y \in Y$
that hits $\pi_{uv,u}$. By the construction of $Z$, the vertex $y$ forces $u \in Z$.
As the choice of $uv$ was arbitrary, $Z$ is a vertex cover of $G_0$ and the theorem is proven.
\maybeqed\end{proof}

We are now ready to the main lower bound construction.

\begin{theorem}[Theorem~\ref{thm:intro:lb:col} restated]\label{thm:lb:col}
Let $H$ be a graph that contains a connected component that is neither a path nor a clique.
Then, unless ETH fails, there does
not exist an algorithm that, given a \cshit{} instance $(G,\col)$
and a tree decomposition of $G$ of width $t$, resolves $(G,\col)$
in time $2^{o(t^{\msep(H)})} |V(G)|^{\Oh(1)}$.
\end{theorem}
\begin{proof}
The case $\msep(H) = 1$ is proven by Theorem~\ref{thm:lb:col-vc}, so in the remainder
of the proof we focus on the case $\msep := \msep(H) \geq 2$.

We show a polynomial-time algorithm that, given a clean $3$-CNF formula $\Phi$ with $n$
variables,
outputs a \cshit{} instance $(G,\col)$ together with a tree decomposition of $G$
of width $\Oh(n^{1/\msep})$ and an integer $k$, such that $\Phi$ is satisfiable
if and only if there exists a set $X \subseteq V(G)$ of size at most $k$
that hits all $\col$-$H$-subgraphs of $G$.
By Lemma~\ref{lem:preprocess}, this would in fact give a reduction from an arbitrary
$3$-CNF formula, and hence conclude the proof of the theorem by Theorem~\ref{thm:spars}.

Let $a,b \in V(H)$ be such that $S$ is a minimal $ab$-separator in $H$ and $|S| = \msep$.
We pick vertices $a$ and $b$ in such a manner that neither of them is a cutvertex in $H$;
observe that this is always possible.
Let $A,B$ be the connected components of $H \setminus S$ that contain $a$ and $b$, respectively.
Finally, let $D$ be the connected component of $H$ that contains both $a$ and $b$.

We first develop two auxiliary gadgets for the construction.
The first gadget, an $\alpha$-\emph{OR-gadget} for $\alpha \in \{a,b\}$, is constructed as follows.
Let $c$ and $d$ be two arbitrary vertices of $S$.
We take three copies $D^1,D^\circ,D^2$ of $H[D]$ and identify:
\begin{enumerate}
\item the vertex $c$ in the copies $D^1$ and $D^\circ$, and
\item the vertex $d$ in the copies $D^2$ and $D^\circ$.
\end{enumerate}
The vertices $\alpha$ (recall $\alpha \in \{a,b\}$)
  in the copies $D^1$ and $D^2$ are called the \emph{attachment points}
of the OR-gadget; let us denote them by $\alpha^1$ and $\alpha^2$, respectively.
For any graph $G$ and coloring $\col: V(G) \to V(H)$,
and for any two vertices $u,v \in V(G)$ of color $\alpha$, by \emph{attaching an OR-gadget}
to $u$ and $v$ we mean the following operation: we create a new copy of the $\alpha$-OR-gadget,
 and identify the attachment vertices $\alpha^1$ and $\alpha^2$ with $u$ and $v$, respectively.
The following claim summarizes the properties of an $\alpha$-OR-gadget.
\begin{myclaim}\label{cl:OR-gadget}
Let $(G',\col')$ be a colored graph created by attaching an $\alpha$-OR-gadget
to $u$ and $v$ in a colored graph $(G,\col)$. Let $\Gamma$ be the vertex set of the gadget
(including $u$ and $v$).
Then
\begin{enumerate}
\item any set $X \subseteq V(G')$ that hits all $\col'$-$H[D]$-subgraphs of $G'$
  needs to contain at least two vertices of $\Gamma$, including at least one vertex
  of $\Gamma \setminus \{u,v\}$;
\item there exist sets $X^u, X^v \subseteq \Gamma$, each of size $2$, such that $u \in X^u$,
  $v \in X^v$, and both these sets hit all $\col'$-$H[D]$-subgraphs of $G'$
  that contain at least one vertex of $\Gamma \setminus \{u,v\}$.
\end{enumerate}
\end{myclaim}
\begin{proof}
For the first claim, observe that $X$ needs to contain a vertex $x \in D^\circ \subseteq \Gamma \setminus \{u,v\}$.
If $x$ is not equal to $c$ in the copy $D^\circ$, then $X$ needs to additionally contain a vertex
of $D^1$. Symmetrically, if $x$ is not equal to $d$ in the copy $D^\circ$, then
$X$ needs to additionally contain a vertex of $D^2$.

For the second claim, let $X^u$ consist of $u$ and the vertex $d$ in the copy $D^\circ$,
and let $X^v$ consist of $v$ and the vertex $c$ in the copy $D^\circ$.
Let $\pi$ be any $\col'$-$H[D]$-subgraph of $G'$ such that $\pi(D)$ contains
a vertex of $\Gamma \setminus \{u,v\}$.
We argue that $\pi$ is hit by $X^u$; the argumentation for $X^v$ is symmetrical.
Since neither $a$ nor $b$ is a cutvertex of $H$, and both
$u$ and $v$ have the same color in $\col$, we have $\pi(D) \subseteq \Gamma$.
If $\pi(d) \in D^\circ$ then we are done, so $\pi(d) \in D^1$. 
Since $H[D]$ is connected, $\pi(c) \in D^1$. Moreover, we now have two options for $\pi(a)$: either $u$, or the vertex $a$ in the copy $D^\circ$. If it was not that $\pi(a)=u$, then $c$ would be a cut-vertex in $H$ that would separate $a$ from $b$. However, since $S$ contains also $d$ which is different than $c$, this would be a contradiction with $S$ being minimal. 
Consequently, $\pi(a)=u$ and $X^u$ hits $\pi$.
\cqed\end{proof}

The second gadget, an $\alpha$-$r$-cycle for $\alpha \in \{a,b\}$ and integer $r \geq 2$,
is constructed as follows.
We first take $r$ vertices $\alpha^1,\alpha^2,\ldots,\alpha^r$, each colored $\alpha$.
Then, we attach an $\alpha$-OR-gadget to the pair $\alpha^i,\alpha^{i+1}$ for every $1 \leq i \leq r$
(with the convention $\alpha^{r+1} = \alpha^1$).
For any graph $G$ and coloring $\col: V(G) \to V(H)$,
and for any sequence of pairwise distinct vertices $u^1,u^2,\ldots,u^r$, each colored $\alpha$,
by \emph{attaching an $\alpha$-$r$-cycle} to $u^1,u^2,\ldots,u^r$ we mean the following
  operation: we create a new copy of the $\alpha$-$r$-cycle and identify $u^i$
  with $\alpha^i$ for every $1 \leq i \leq r$.
The following claim summarizes the properties of an $\alpha$-$r$-cycle.
\begin{myclaim}\label{cl:cycle}
Let $(G',\col')$ be a colored graph created by attaching a $\alpha$-$r$-cycle
to $u^1,u^2,\ldots,u^r$ in a colored graph $(G,\col)$.
Let $\Gamma$ be the vertex set of the gadget (including all vertices $u^i$).
Then
\begin{enumerate}
\item any set $X \subseteq V(G')$ that hits all $\col'$-$H[D]$-subgraphs of $G'$
  needs to contain at least $r + \lceil r/2 \rceil$ vertices of $\Gamma$,
        including at least $r$ vertices of $\Gamma \setminus \{u^i: 1 \leq i \leq r\}$;
\item for every set $I \subseteq \{1,2,\ldots,r\}$ that contains either $i$ or $i+1$
  for every $1 \leq i < r$, and contains either $1$ or $r$, 
  there exist sets $X^I \subseteq \Gamma$ of size $|I|+r$ such that $u^i \in X^I$ whenever $i \in I$,
  and $X^I$ hits all $\col'$-$H[D]$-subgraphs of $G'$
  that contain at least one vertex of $\Gamma \setminus \{u^i: 1 \leq i \leq r\}$.
\end{enumerate}
\end{myclaim}
\begin{proof}
For the first claim, apply the first part of Claim~\ref{cl:OR-gadget}
to each introduced $\alpha$-OR-gadget:
$X$ needs to contain at least one vertex that is not an attachment vertex
in each $\alpha$-OR-gadget between $u^i$ and $u^{i+1}$ ($r$ vertices in total)
and, moreover, at least two vertices in each $\alpha$-OR-gadget.
For the second claim, construct $X^I$ as follows: first take all vertices $u^i$ for $i \in I$
and then, for each $1 \leq i \leq r$, insert into $X^I$ the set $X^{u^i}$
from Claim~\ref{cl:OR-gadget}, if $i \in I$, and the set $X^{u^{i+1}}$ otherwise.
Observe that, in the second step, each index $i$ gives rise to exactly one new vertex
of $X^I$, and hence $|X^I| = |I| + r$. Moreover, the required hitting property of $X^I$
follows directly from Claim~\ref{cl:OR-gadget}.
\cqed\end{proof}

Armed with the aforementioned gadgets,
we now proceed to the construction of the instance $(G,\col)$.
Let $s$ be the smallest positive integer such that $s^\msep \geq 3n$. Observe that
$s = \Oh(n^{1/\msep})$. We start our construction by introducing
a set $M$ of $s\msep$ vertices $w_{i,c}$, $1 \leq i \leq s$, $c \in S$.
We define $\col(w_{i,c}) = c$.
The set $M$ is the central part of the constructed graph $G$.
In particular, 
in our reduction each connected component of $G \setminus M$ will be of constant size,
yielding immediately the promised tree decomposition.

To each clause $C$ of $\Phi$, and to each literal $l$ in $C$, assign a function
$f_{C,l}: S \to \{1,2,\ldots,s\}$ such that $f_{C,l} \neq f_{C',l'}$ for $(C,l) \neq (C',l')$.
Observe that this is possible due to the assumption $s^\msep \geq 3n$ and the fact
that $\Phi$ is clean.

For each variable $x$ of $\Phi$, proceed as follows.
First, for each clause $C$ and literal $l \in \{x, \neg x\}$,
we introduce a copy $D_{x,C,l}$ of $H[N[A]]$ and identify
every vertex $c \in S$ in the copy $D_{x,C,l}$ with the vertex $w_{f_{C,l}(c),c}$.
Let $a_{x,C,l}$ be the vertex $a$ in the copy $D_{x,C,l}$.
Second, introduce a new dummy vertex $a_x$, colored $a$.
Finally, attach an $a$-$4$-cycle to vertices $a_{x,C_1,l}, a_{x,C_2,\neg l}, a_{x,C_3,l}, a_x$;
recall that $x$ appears exactly three times in $\Phi$, twice positively and once
negatively or twice negatively and once positively.

For each clause $C$ of $\Phi$, proceed as follows. 
First, for each literal $l$ in $C$ introduce a copy $D_{C,l}$ of $H[N[B]]$ and identify
every vertex $c \in S$ in the copy $D_{C,l}$ with the vertex $w_{f_{C,l}(c),c}$.
Let $b_{C,l}$ be the vertex $b$ in the copy $D_{C,l}$.
Second,
  we attach a $b$-$r_C$-cycle on the vertices $b_{C,l}$, where $2 \leq r_C \leq 3$ is the number
of literals in $C$.

We define $k = 12n-m$, where $n$ is the number of variables in $\Phi$ and $m$ is the number
of clauses.

Finally, perform the following two operations.
First, for each connected component $L$ of $H[D] \setminus S$ that is different than $A$ and $B$,
and for each function $f: N(L) \to \{1,2,\ldots,s\}$, create a copy $D_{L,f}$ of $H[N[L]]$ and
identify each vertex $c \in N(L) \subseteq S$ with the vertex $w_{f(c),c}$.
Let $(G',\col')$ be the colored graph constructed so far.
Second, introduce a large number (at least $k+1$) of disjoint copies of $H \setminus D$
into the graph $G$. This concludes the construction of the \cshit{} instance $(G,\col)$.

Observe that each connected component of $G \setminus M$ is of constant size.
Thus, it is straightforward to provide a tree decomposition of $G$ of width
$s\msep + \Oh(1)=  \Oh(n^{1/\msep})$.
Hence, it remains to argue about the correctness of the construction.

In one direction, let $\phi$ be a satisfying assignment of $\Phi$.
Define a set $X \subseteq V(G)$ as follows.
\begin{enumerate}
\item For each variable $x$, include into $X$
  the set $X^I$ from Claim~\ref{cl:cycle} for the $a$-$4$-cycle created for the variable $x$
  that contains the vertices $a_{x,C,l}$ for clauses $C$ where $l$ is evaluated
  to true by $\phi$, 
  and the vertex $a_x$ if there is only one such clause.
  Note that the construction of the attachment points of the $a$-$4$-cycle created for $x$
  ensures that $X^I$ contains exactly two non-consecutive attachment points, and hence
  $|X^I| = 6$.
\item For each clause $C$, pick one literal $l$ that is satisfied by $\phi$
in $C$, and include into $X$ the set $X^I$ from Claim~\ref{cl:cycle}
for the $b$-$r_C$-cycle created for $C$ that contains $b_{C,l'}$ for all $l' \neq l$.
  Observe that we include exactly $2r^C-1$ vertices to $X$ for clause $C$.
\end{enumerate}
Since $\Phi$ is clean, we observe that
\begin{equation}\label{eq:lb-col}
|X| = 6n + \sum_{\mathrm{clause\ }C} (2r_C-1) = 6n + 2\cdot 3n - m = 12n-m = k.
\end{equation}
Consider any $\col$-$H$-subgraph $\pi$ of $G$.
If $\pi(V(H))$ contains a vertex of some $\alpha$-$r$-cycle that is not an attachment vertex,
then $\pi$ is hit by $X$ by Claim~\ref{cl:cycle}.
Otherwise, observe that the color constraints imply that $\pi(S) \subseteq M$.
Consequently, there exists $f: S \to \{1,2,\ldots,s\}$ such that
$\pi(c) = w_{f(c),c}$ for each $c \in S$.
The only vertices of $G$ that are both adjacent to $M$ and have colors from $A$
belong to the copies $D_{x,C,l}$. Similarly, the only vertices
of $G$ that are both adjacent to $M$ and have colors from $B$ belong
to the copies $D_{C,l}$. 
As $N_H(A) = N_H(B) = S$ and both $H[A]$ and $H[B]$ are connected, we infer that there exists a clause $C$ and literal $l \in C$
corresponding to a variable $x$,
such that $f = f_{C,l}$, and $\pi$ maps $A$ to $D_{x,C,l}$ and $B$ to $D_{C,l}$.
However, observe that
if $l$ is satisfied by $\phi$, then $X$ contains the vertex $a_{x,C,l} \in D_{x,C,l}$,
  and otherwise $l$ does not satisfy $C$ and $X$ contains $b_{C,l} \in D_{C,l}$.
Consequently, $\pi$ is hit by $X$.
As the choice of $\pi$ was arbitrary, we infer that $X$ hits all $\col$-$H$-subgraphs
of $G$.

In the other direction, let $X$ be a set of at most $k$ vertices of $G$
that hits all $\col$-$H$-subgraphs.
As $G$ contains at least $k+1$ disjoint copies of $H \setminus D$, we infer that
$X \cap V(G')$ hits all $\col'$-$H[D]$-subraphs of $G'$.
By Claim~\ref{cl:cycle}, $X$ contains at least $6$ vertices of each $a$-$4$-cycle introduced
for every variable $x$, and at least $2r_C-1$ vertices of each $b$-$r_C$-cycle
introduced for every clause $C$ (because $2r_c-1=r_c+\lceil r_c/2\rceil$ for $2\leq r_c\leq 3$). However, as these gadgets
are vertex disjoint, by similar calculations as in~\eqref{eq:lb-col} we infer that
these numbers are tight: $X$ contains \emph{exactly} $6$ vertices in each $a$-$4$-cycle,
\emph{exactly} $2r_C-1$ vertices in each $b$-$r_C$-cycle
  and no more vertices of $G$.
In particular, for each variable $x$, $X$ contains either all vertices
$a_{x,C,x}$ for clauses $C$ where $x$ appears positively,
or all vertices $a_{x,C,\neg x}$ for clauses $C$ where $x$ appears negatively.
Define an assingment $\phi$ as follows: for every variable $x$,
we set $\phi(x)$ to true if $X$ contains all vertices $a_{x,C,x}$,
and to false otherwise. We claim that $\phi$ satisfies $\Phi$.

To this end, consider a clause $C$. As $X$ contains only $2r_C-1$ vertices in the $b$-$r_C$-cycle
constructed for $C$, by Claim~\ref{cl:cycle} there exists a literal $l \in C$ such that
$b_{C,l} \notin X$. Let $x$ be the variable of $l$.
Let us construct a $\col$-$H$-subgraph $\pi$ of $G$ as follows:
\begin{enumerate}
\item $\pi|_S = f_{C,l}$;
\item $\pi|_{N_H[A]}$ maps $N_H[A]$ to $D_{x,C,l}$;
\item $\pi|_{N_H[B]}$ maps $N_H[B]$ to $D_{C,l}$;
\item for every component $L$ of $H[D] \setminus S$ that is not equal to $A$ nor $B$,
  $\pi|_{N_H[L]}$ maps $N_H[L]$ to $D_{L,f_{C,l}|_{N_H(L)}}$;
\item $\pi|_{V(H) \setminus D}$ maps $H \setminus D$ to any its copy in $G$.
\end{enumerate}
It is straightforward to verify that $\pi$ is a $\col$-$H$-subgraph of $G$.
Moreover, $\pi(V(H))$ contains only two vertices of the introduced $\alpha$-$r$-cycles:
$a_{x,C,l}$ and $b_{C,l}$. 
Since $X$ cannot contain any vertex outside these $\alpha$-$r$-cycles,
we infer that $a_{x,C,l} \in X$.
Consequently, $\phi$ sets $l$ to true, and thus satisfies $C$.
This finishes the proof of the correctness of the reduction,
and concludes the proof of the theorem.
\maybeqed\end{proof}

\section{Conclusions and open problems}\label{sec:conc}

Our preliminary study of the treewidth parameterization of the \shit{}
problem revealed that its parameterized complexity is highly involved. Whereas for the more graspable colored version we obtained essentially tight bounds, a large gap between lower and upper bounds remains for the standard version.
In particular, the following two questions arise:
Can we improve the running time of Theorem~\ref{thm:std:algo} to 
factor $t^{\msep(H)}$ in the exponent?
Is there any relatively general symmetry-breaking assumption on $H$
that would allow us to show a $2^{o(t^{\msep(H)})}$ lower bound
in the absence of colors?

In a broader view, let us remark that 
the complexity of the treewidth parameterization of \emph{minor-hitting} problems
is also currently highly unclear. Here, for a minor-closed graph class $\mathcal{G}$
and input graph $G$, we seek for the minimum size of a set $X \subseteq V(G)$ such that
$G \setminus X \in \mathcal{G}$, or, equivalently, $X$ hits all
minimal forbidden minors of $\mathcal{G}$.
A straightforward dynamic programming algorithm has double-exponential dependency on the
width of the decomposition.
%, since there are exponentially many possible ways in which a model
%of a forbidden minor can traverse a bag of the decomposition.
However, it was recently shown that $\mathcal{G}$ being the class
of planar graphs, a $2^{\Oh(t \log t)} |V(G)|$-time algorithm exists~\cite{planarization}.
Can this result be generalized to more graph classes?
%The essence of the difficulty here seems to lie in the fact that we do not know any
%good method of certifying that a given graph does not contain some forbidden minor.

\bibliographystyle{plain} 
\bibliography{references}

\end{document}